\documentclass[manyauthors]{fundam}

\usepackage{headers}
\usepackage{macro2}
\usepackage[noend]{algpseudocode}
\usepackage{algorithm}

\begin{document}

\setcounter{page}{1}
\issue{}

\title{Activity Networks with Delays\\ An application to toxicity analysis}

\author{%
 Franck Delaplace \\
 Universit\'e d'Evry - Val d'Essonne, IBISC, France  \and
  Cinzia Di Giusto \corresponding \\
  Universit\'e C\^ote d'Azur, CNRS, I3S, France \and
  Jean-Louis Giavitto \\
   CNRS, IRCAM, UPMC - UMR 9912 STMS \\ INRIA MuTAnt team, Paris, France \and
  Hanna Klaudel \\
   Universit\'e d'Evry - Val d'Essonne, IBISC, France 
}

\address{cinzia.di-giusto@unice.fr}

\runninghead{F. Delaplace et al. }{Activity Networks with Delay}

\maketitle 

\begin{abstract}
      \andy, Activity Networks with Delays, is a discrete time framework aimed at the qualitative modelling of time-dependent activities. The modular and concise syntax  makes \andy suitable for an easy and natural modelling of  time-dependent biological systems (\ie regulatory pathways). 
     
      Activities involve entities playing the role of
      activators, inhibitors or products of biochemical network operation. 
      Activities may have given duration, \ie the time required to obtain results. 
      An entity  may represent an object (\eg an agent, a biochemical species or a family of thereof) with a local attribute, a state denoting  its level (\eg concentration, strength). 
      Entities levels may change as a result of an activity  or may decay gradually as time passes by. 
      
      The semantics of \andy is formally given via  high-level Petri nets ensuring this way some modularity.
      As main results we show that \andy systems have finite state representations even for potentially infinite processes and it well adapts to the modelling of toxic behaviours. As an illustration, we present a classification of toxicity properties and give some hints on how they can be verified with existing tools on \andy systems. 
      A small case study on blood glucose regulation is provided to exemplify the \andy framework  and the toxicity properties.
      
\end{abstract}


\section{Introduction}

Activities and their causal relationships are a central concern in biological systems such as biological networks and regulatory pathways.
They deal with questions such as,  how long does it take to complete an activity, when it can be performed,  
 or whether it can be delayed or even ignored. 
Our proposal stems from reaction systems 
\cite{DBLP:journals/ijfcs/BrijderEMR11}, a formalism based on  \emph{reactions},
each defined as a triple $(R, I, P)$ with $R$  set of
\emph{reactants}, $I$  set of \emph{inhibitors} and $P$  set of
\emph{products}, and $R,I$ and $P$  taken from a common set of
\emph{species}. Reaction systems are based on three basic assumptions:
\begin{enumerate}[(i)]
 \item \label{item:react} 
 a reaction can take place only if all the reactants involved are available but none of the inhibitors are; 
 \item \label{item:quantity}  if a species is available then a sufficient amount of it is necessary to trigger a reaction; 
 \item  species are not persistent: they become unavailable if they are not sustained by a reaction.
\end{enumerate}
We complement this model by adding several important features.
We allow a richer description of species states using attributes having potentially multiple values  and introduce timing aspects.  
More precisely, the resulting formalism, \emph{Activity Networks with Delays} (\andy), 
targets systems 
composed of various \emph{entities} (\ie species in reaction systems) that evolve by means of time-dependent activities.
Each entity is characterised by one attribute,
called \emph{level}, representing rates, activation/inhibition status or expression degrees
(\eg low, medium, high). 
\emph{Activities} are rules describing the evolution of systems, involving (as for reaction systems) three sets of entities:
\emph{activators}, which must be present, 
\emph{inhibitors}, which must be absent,  
and \emph{results} whose expression levels have to be modified (increased or decreased).
The introduction of time concerns both activities and entities. Activities have a duration and entities are subject
to a \emph{decay} process or aging that models the action of a
non-specified environment: \ie levels decrease with time
progression.
Another difference is in the semantics of activities: while in reaction systems a maximal concurrency model is considered, here we  
adopt two types of activities:  \emph{mandatory} and \emph{potential} ones. The former set of activities, once enabled, must take place altogether in the same time unit (in maximal concurrency) while the latter, non deterministically one at the time, may be performed or not.  
Following \cite{Wang04formalverification}, as in our modelling all entities have discrete states and share the same global clock, 
it is reasonable to work with discrete time constraints.\looseness=-1

 Our main objective is to provide theoretical foundations and underlying tools for the
 description and  the understanding of the mechanisms and of the
 structural properties underpinning biological interactions networks. In
 biology, ordinary differential equations (ODE) remain the
 predominant modelling methodology\cite{alon2006introduction}. Such models present however some
 drawbacks: they need a precise quantification of parameters, 
random interactions are difficult to model outside an averaged
 approach, %
 system openness and non-linearities make hopeless the existence of
 analytic solutions, %
 and numerical methods -- often the only sensible option -- are mainly
 descriptive and cannot form the basis for an explanatory theory. 
 On the other hand, many
 languages and tools  have been developed  to build and explore discrete qualitative
 models in various application domains. Such models are rough
 abstractions of real world processes but they make possible to unravel
 the entangled causal relationships between system's entities.
 More specifically in systems biology, formalisms like Petri
 nets~\cite{DBLP:journals/nc/BaldanCMS10}, Boolean networks~\cite{thomas1973boolean},
 process algebras~\cite{journals/tcsb/Cardelli05} 
or rewriting
 systems~\cite{giavitto04a} (to cite a few) have been used with
 undeniable success to understand the causal links between structure
 and behaviour in molecular networks.
  Nonetheless, the management of time is somewhat problematic in all of
 the previous approaches. 
 In general, timing aspects are either disregarded or handled at a primitive level,
 which leads to expressiveness problems for the modeller. 
 For example, modelling a synchronous evolution directly in Petri nets
 requires a coding that is not necessarily natural for a biologist.
 This is why here we propose  a qualitative discrete formalism for biological systems that 
 provides, in particular,  a direct and intuitive account of timing aspects natural in biology such as activity duration and decay. Such features are in particular essential in describing toxicity problems.
 Indeed toxicology~\cite{Waters2004} studies the adverse effects of the exposures to chemicals at various levels of living entities: organism, tissue, cell and intracellular molecular systems. During the last decade, the accumulation of genomic and post-genomic data together with the introduction of new technologies for gene analysis has opened the way to  \emph{toxicogenomics}. Toxicogenomics combines toxicology with
``Omics'' technologies\footnote{``Omics'' technologies are methodologies such as genomics, transcriptomics, proteomics and metabolomics.}
to study the \emph{mode-of-action} of toxicants or environmental stressors on biological systems. The mode-of-action is understood as the sequence of events from the absorption of chemicals to a toxic outcome.
Toxicogenomics potentially improves clinical diagnosis capabilities and facilitates the identification of potential toxicity  in drug discovery~\cite{Foster2007} or in the design of bio-synthetic entities~\cite{synthetic}. 
 Our formalism permits to model biological systems together with their stressors and discover (or highlight) the mode-of-action of toxicants, that is why  we complement our designing process with a general discussion on toxicity properties and how they can be tested against \andy systems. In this respect, we aim at providing a methodology rather than a precise technique of testing.

\paragraph{\bf Organisation of the paper.}  This paper is the extended version of BioPPN workshop paper \cite{CDGHKFD14}. The main difference with respect to \cite{CDGHKFD14} is the introduction of time aspects in \andy systems, a more mature taxonomy of toxicity properties and a prototype implementation of \andy networks using Snakes toolkit \cite{Pom-PNN-2008},  Snoopy~\cite{DBLP:conf/apn/HeinerHLRS12} and its related analysis tool Charlie \cite{HSW15}.

Section \ref{sec:reaction} introduces the principles behind \andy
networks
and gives a formal definition of \andy semantics  in terms of high-level Petri nets. It also states the main result of the paper  addressing the finiteness of state space.
Next, Section \ref{sec:toxic} discusses an axiomatisation of toxicity properties. An illustration of the model and its properties is given in Section \ref{sec:example} describing the assimilation of aspartame in  blood regulation in human body. Finally Sections \ref{sec:related} and \ref{sec:concl} discuss related works and conclude.

A web page with additional material is available at \cite{webpage}. It includes the implementations in Snakes and Snoopy and  a technical report adding a comparison with timed automata.


\section{Activity Networks with Delays}\label{sec:reaction}

In this section we give the syntax and   
semantics of \emph{Activity Networks with Delays} (\andy). An \andy is composed by a set of entities
\entities driven by a set of timed activities. The time is assumed to be  discrete and modelled by a tick event.

\paragraph{\bf Entities.} 
Each entity, $\entity \in \entities$, is associated to a finite number of levels $\setlev_{\entity}$ (from 0 to $\setlev_{\entity}-1$) that 
in general represent ordered expression degrees (\eg low, medium,
high) and refer to a change in the entity capability of action. A decay duration is associated to each level of an entity $e$ 
through the function $\life_\entity:[0 \mydots \setlev_\entity-1] \to \nat^+ \cup \{\omega \}$ 
allowing different duration's for different levels or  $\omega$ if  unbounded. 
More precisely, if $\life_\entity(i) = \omega$, then level $i$ is \emph{permanent} and it can be modified
only by an activity. If $\life_\entity(i) \neq \omega$, the presence of the entity
at level $i$ is \emph{transient}: once entity $\entity$ at level $i$ has passed $\life_\entity(i)$ time units its level will pass to $i-1$, \ie it decays or ages gradually
as time passes by.
We fix $\life_{\entity}(0) = \omega$ as no negative level is allowed. 
 In principle, this substitutes a set of unspecified activities accounting for the action of an underspecified   
environment that consumes entities.  
A state $\state$ of an \andy network assigns to each entity $\entity \in \entities$ a level $\eta \in [0..\setlev_{\entity}-1]$. The initial state $\state_0$ sets each entity $\entity$ to a given level $\level{\entity}$.

\paragraph{\bf Activities.}
The evolution of entities is driven by timed \emph{activities}  of the form:
 $$ \rho ::= \activ{A_{\rho}}{I_{\rho}}{\dur{\rho}}{R_{\rho}}$$  
where $\dur{\rho} \in \nat$ is the activity duration, $A_{\rho}$ (activators) and $I_{\rho}$
(inhibitors) are sets of pairs $(\entity, \level{\entity})$ with $\entity\in
\entities$ and $\level{\entity} \in [0 \mydots \setlev_\entity-1]$; $R_{\rho}$, the 
results, is a non empty set of pairs $(\entity, \pm n)$ where $\entity\in
\entities$ and $ \pm n \in
\mathbb{Z}$ is a positive or negative variation of the entity level $\level{\entity}$. 
Entities can appear at most once in each set $A_{\rho}$, $I_{\rho}$ and $R_{\rho}$, cf., Remark \ref{rem:def} below. 
We write $e \in A_{\rho}$ to denote  $(e, \cdot) \in A_{\rho}$, similarly for $I_{\rho}$ and $R_{\rho}$. We omit index $\rho$ if it is clear from the context.

 Duration $\dur{\rho}$ characterises the  number of ticks required for yielding change of
levels of results, $\dur{\rho} =0 $ denotes an instantaneous activity. The results of an activity are the increments or decrements  $\pm n$ of each entity level  in $R_{\rho}$ (up to the range of entity levels). 
An activity $\rho$  can take place only if it is \emph{enabled}, this depends on
 the activator and inhibitor
levels appearing in  $A_{\rho}$ and $I_{\rho}$. 
Nonetheless, observing only the current level of each entity is not enough to
decide if an activity is enabled: such approach would miss the handling
of decays and durations. 
\begin{definition}[Enabled activity]\label{def:enabledactivity}
An activity $\rho$ is
enabled (\ie may happen) if and only if:
\begin{itemize}
 \item for each entity $\entity_a$ of the activators set $A_{\rho}$ (\ie $(\entity_a, \level{a}) \in A_{\rho}$), $\entity_a$ is available at least at level $\level{a}$ for
  the whole duration $\dur{\rho}$;
 \item for each entity $\entity_i$ of the inhibitors set $I_{\rho}$ (\ie $(\entity_i, \level{i}) \in I_{\rho}$), $\entity_i$ is available at a level strictly inferior to $\level{i}$ for
  the whole duration $\dur{\rho}$;
\end{itemize}

\end{definition}

\begin{remark} \label{rem:def}
Notice that:
\begin{itemize}
\item If an entity is an activator (resp. an inhibitor) at level
      $\ell$ for some activity $\rho$, it is also an activator (resp. an
      inhibitor) for $\rho$ at all levels $l \ge \ell$ (resp. $l \le \ell$).

 \item  An entity may be both in $A$ and $R$, \eg as in an
      auto-catalytic production where reaction products themselves are activators for the reaction. Similarly, an
      entity may be both in $I$ and $R$: thus representing self-repressing activities.

\item An entity can appear in the same activity simultaneously as an
      activator and as an inhibitor but we require them to
      occur with different levels. For instance, the rule
      $$
      \activ{\{({\entity},\level{})\} \cup A}{\{({\entity}, \level{}')\} \cup I}{\dur{}}{R}
      \qquad  \text{where~}   \level{} < \level{}'
      $$
      requires that the activity takes place only if the level
      $\lev_{\entity}$ of ${\entity}$ belongs to the interval $\level{}
      \leq \lev_{\entity} < \level{}'$. In particular, if ${\entity}$
      has to be present in an activity exactly at level $\level{}$,
      ${\entity}$ should appear as an activator at level $\level{}$ and
      as inhibitor at level $\level{}' = \level{}+1$.

\item The set of activators and inhibitors $A \cup I $ is allowed to be
      empty. In this case, the activity is constantly enabled. This
      accounts for modelling an environment that  continuously
      sustains the production of an entity. In this case a duration
      $\dur{} \not= 0$ may be used to account for a start-up time.

\end{itemize}
 
\end{remark}

Guided by the case study and the biological applications, we assume that once an activity $\rho$ is triggered:

\begin{itemize}
 
\item it must wait at least $\dur{\rho}$ before being enabled again (a kind of  refractory period).
     
\item its activators and inhibitors are left unchanged.

\end{itemize}

These are non restrictive hypotheses and by slightly modifying Definition \ref{def:rs} one can obtain different semantics that could depend on the specific modelled scenario: \eg once enabled an activity $\rho$ can be triggered an unbounded number of times or  once per time unit; also one can easily choose a different policy to modify the level of activators and/or inhibitors.

It is important to note that when an activity is enabled, it is not
necessarily performed. The idea is to make a distinction between
two types of activities: 
\begin{enumerate}
\item \emph{\potential} activities, denoted by $\alpha$ in a set $\Reac$:
      if enabled, they \emph{may} occur in an interleaved way thus the result
      depends on the order they are performed. It may happen that no \potential activities occur even if they are enabled;
\item \emph{\obligatory} activities, denoted by $\beta$ in a set $\Syn$: if
      enabled  they \emph{must} occur within the same time unit. All enabled
      \obligatory activities occur simultaneously.
\end{enumerate}

As all \obligatory activities $\beta_i$ are triggered
simultaneously, all updates related to an entity $\entity$ are
collected and summed up before being applied to the current level of
$\entity$. 

\begin{definition}[\andy network and its semantics]\label{def:andy}
An \emph{\andy network} is a triple $(\entities, \Syn, \Reac)$ where
$\entities$ is the set of entities, $\Syn$ is the set of \obligatory
activities and $\Reac$ is the set of \potential ones.  
It evolves by triggering enabled activities in two separate phases implementing one evolution step:
\begin{enumerate}[Ph. 1]

\item \label{item:reaction} 
      In this phase (between two ticks) \textbf{\potential activities} that are
      enabled may be triggered. Their action occurs non-deterministically in an interleaved
      way. Any potential activity is triggered at most once in this phase. 

\item \label{item:clock} 
      This phase is the tick transition: the effect is that all  concerned entities must decay and  all enabled \textbf{\obligatory
        activities} are performed simultaneously.
      Notice that as activators and inhibitors are not ``consumed", there
      are no conflicts and so there is a unique way (per time unit) to
      execute all enabled \obligatory activities.
\end{enumerate}

\end{definition}

\begin{example}[Repressilator]\label{ex:rep}
The repressilator \cite{Elowitz} is a synthetic genetic regulatory
network in Escherichia Coli consisting of three genes ($lacI, tetR, cI$) connected in a feedback loop such
that each gene inhibits the next gene and is inhibited by the previous
one. The repressilator exhibits a stable oscillation with fixed time
period that is witnessed by the synthesis of a green fluorescent protein ($GFP$).
According to \cite{Elowitz}: \emph{``the resulting oscillation with typical periods of hours are slower than the cell division cycle so the state of the oscillator has to be transmitted from generation to generation''}.  

It may be modelled in \andy with five entities $\entities =\{lacI, tetR, cI, GFP, gen\}$.
The first three represent the three genes,  each with two levels denoting an \emph{on}/\emph{off} state where level
$1=on$ has decay time $2$. $GFP$ has two levels (\emph{on}/\emph{off}) with unbounded decay. 
Finally, the evolution of successive generations is represented by entity $gen$ with, let's say, seven levels  (from 0 to 6, with unbounded decay) to differentiate among generations.  
 
The behaviour is modelled  with \potential and \obligatory activities: \potential activities (on the left below) are used to represent the feedback loop between $lacI, tetR$ and $cI$. \Obligatory activities  (on the right) handle the synthesis of $GFP$ and show the evolution of generations.
\[
\begin{array}{ll}
 \alpha_1: \activ{\emptyset}{(lacI,1)}{2}{(tetR,+1)} \quad  &  \beta_1: \activ{(lacI,1)}{\emptyset}{0}{(GFP, +1)}  \\
 \alpha_2: \activ{\emptyset}{(tetR,1)}{2}{(cI,+1)} &  \beta_2: \activ{(lacI,0)}{(lacI,1)}{0}{(GFP,-1)}\\
 \alpha_3: \activ{\emptyset}{(cI,1)}{2}{(lacI,+1)} &  \beta_3: \activ{(gen,0)}{(gen, 6)}{1}{(gen, +1)}\\
 &  \beta_4: \activ{(gen,6)}{\emptyset}{1}{(gen, -6)} 
\end{array}
\]
In particular, the evolution of successive generations is circular modulo 7,  this allows us to have a finite number of levels for $gen$). Activities $\beta_3$ and $\beta_4$ describe such a behaviour. Notice in particular the use of inhibitor $(gen, 6)$ in $\beta_3$ that allows to apply the activity for all levels of $gen$ from 0 to 5 and to apply $\beta_4$ for level 6.

We now exhibit a possible execution scenario for the repressilator, the table below shows for each line the current level of each entity:
\[
\begin{tabular}{c|c|c|c|c|l}
 $lacI$~ & $tetR$ & $cI$& $GFP$ & $gen$ & description\\
   0    &  0     &    1 &   0   &   0   & ~initial state\\
   0    &  0     &    1 &   0   &   0   & ~after Ph. 2 (tick)\\
   0    &  0     &    1 &   0   &   0   & ~after Ph. 1 (no activity)\\
   0    &  0     &    1 &   0   &   1   &  ~after Ph. 2  (tick and activity $\beta_3$)  \\
   0    &  1     &   1  &   0   &   1   & ~after Ph. 1 (activity $\alpha_1$)\\
   0    &  1     &   1  &   0   &   1   & ~after Ph. 2 (tick)\\ 
   0    &  1     &   1  &   0   &   1   & ~after Ph. 1 (no activity)\\
   0    &  1     &   0  &   0   &   2   & ~after Ph. 2  (tick, decay of $cI$  and activity $\beta_3$)\\
   1    &  1     &   0  &   0   &   2   & ~after Ph. 1 (activity $\alpha_3$)\\    
   1    &  1     &   0  &   0   &   2   & ~after Ph. 2 (tick)\\ 
   1    &  1     &   0  &   0   &   2   & ~after Ph. 1 (no activity)\\ 
   1    &  0     &   0  &   1   &   3   & ~after Ph.~2 (tick, decay of $tetR$ and activities $\beta_1,\beta_3$)
\end{tabular}
\]

\hfill $\diamond$
\end{example}


\subsection{\bf Petri Net formalisation.}\label{sec:petri} 

The semantics of \andy networks is formalised in terms of high-level Petri nets. 
We recall here the  notations together with some elements of their
semantics  \cite{DBLP:series/eatcs/Jensen92}.

\begin{definition}\label{def:PTS} 
A {\bf high-level Petri net}  is a tuple $(P, T, F, L, M_0)$ where:
\begin{itemize}
 \item $P$ is the set of places, $T$ is the set of transitions, with $P \cap T = \emptyset$;
 \item $F \subseteq (P \times T) \cup (T \times P)$ is the set of arcs; 
 \item $L$ is the labelling  of $P \cup T\cup F$ defined as follows:
\renewcommand{\labelitemii}{$-$} 
\begin{itemize}
 \item $\forall p \in P$, $L(p) $ is a set of values (the type of $p$);
 \item $\forall t \in T$, $L(t)$ is a  computable Boolean expression (a guard);
 \item and $\forall f\in F$, $L(f)$  is a  tuple of variables or values.
\end{itemize}
 \item $M_0$ is the initial marking putting a multiset of values (tokens) in $L(p)$ in each   $p \in P$.\looseness=-1
\end{itemize} 
\end{definition}
The behaviour of high-level Petri nets is defined as usual: markings (states) are functions from places in $P$ to multisets of possibly structured tokens in $L(p)$. A transition $t \in T$ is \emph{enabled} at  marking $M$, if there exists a valuation $\sigma$ of all variables in the labelling of $t$  such that the guard $L(t)$ evaluates to true ($L_{\sigma}(t) = \mathsf{true}$) and there are enough tokens in all input places $p$ of $t$ to satisfy the corresponding input arcs, \ie $L_{\sigma}((p,t)) \in M(p) $ for each input place $p$ of $t$.
The firing of $t$ produces the marking $M'$:  
$\forall p \in P, M'(p) = M(p) - L_{\sigma}((p,t)) + L_{\sigma}((t,p))$
with $L_{\sigma}(f) = 0 $ if $f\notin F$, $-$ and $ +$ are  multiset  operators for removal and adding of one element, respectively. We denote it by $\firing{M}{t}{\sigma}{M'}$. 
Observe that for \andy we are considering a subclass of high-level Petri nets where for each pair place/transition there is either no connection or a loop (an input and an output arc). So, starting from an initial marking associating one token per place, the number of tokens (but not necessary their value) is maintained at each step of the net evolution.

As usual, in figures, places are represented as rounds, transitions as squares, 
and arcs as directed arrows.
By convention, primed version of variables (\eg $x'$) are used to annotate 
output arcs of transitions, their evaluation is  
possibly computed using unprimed 
variables (\eg $x$) appearing on input arcs. 
With an abuse of notation, singleton markings are denoted without brackets, 
the same is used in  arc annotations. 
Also we refer to valuated variables without effectively mentioning the valuation $\sigma$: \eg we say that the current value of variable $w$ is $w$ instead of $\sigma(w)$.
An example of firing is shown in  Figure \ref{fig:firing}.

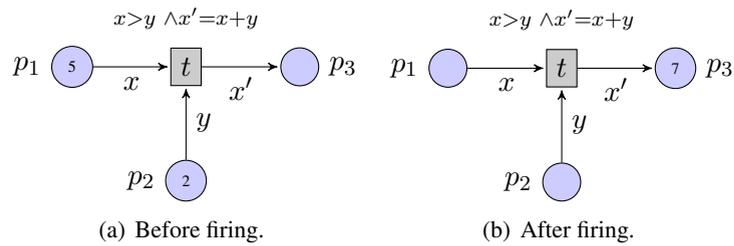
\begin{figure}[h]
\begin{center}
\subfigure[Before firing.\label{fig1}]{
\begin{tikzpicture}[node distance=1.5cm,>=stealth',bend angle=45,auto]

  \node [place, label = left: $p_1$] at (-1.5,0) (Q1){\tiny{5}};
  \node [place, label=left:$p_2$]at (0,-1.5)  (Q2){\tiny{2}};
  \node [place, label = right: $p_3$] at (1.5,0) (Q3){};
  
  \node [transition] (t) [  label=above : {$\begin{array}{c}
                                                {\scriptstyle x>y ~\wedge x'=x+y}
                                               \end{array}$}] at (0,0) {$t$}
   edge [pre]    node[below] {$x$}    (Q1)
   edge [pre]    node[right] {$y$}    (Q2)
   edge [post]   node[below]{$x'$}    (Q3);
\end{tikzpicture}}
\subfigure[After firing.\label{fig2}]{
\begin{tikzpicture}[node distance=1.5cm,>=stealth',bend angle=45,auto]

  \node [place, label = left: $p_1$] at (-1.5,0) (Q1){};
  \node [place, label=left:$p_2$]at (0,-1.5)  (Q2){};
  \node [place, label = right: $p_3$] at (1.5,0) (Q3){\tiny{7}};
  
  \node [transition] (t) [  label=above:{ $\begin{array}{c}
                                                {\scriptstyle x>y ~\wedge x'=x+y}
                                               \end{array}$}] at (0,0) {$t$}
   edge [pre]    node[below] {$x$}    (Q1)
   edge [pre]    node[right] {$y$}    (Q2)
   edge [post]   node[below]{$x'$}    (Q3);
\end{tikzpicture}}
\end{center}
\caption{\label{fig:firing}Example of firing of transition $t$ with $\sigma = \{ x= 5, y=2, x'=7\}$.}
\end{figure}

We say that a marking $M$ is \emph{reachable} from the initial marking 
$M_0$ if there exists a firing sequence $(t_1, \sigma_1), \mydots, (t_n, \sigma_n)$ such that 
 $\firing{M_0}{t_1}{\sigma_1}{M_1} \mydots \firing{M_{n-1}}{t_n}{\sigma_n}{M}$.
The semantics of  an initially marked Petri net is a marking graph (transition system) 
comprising all the reachable markings.

\paragraph{\bf \andy's formalisation.} 
 \andy syntax and semantics of Definition \ref{def:andy} are compact and intelligible, thus easily usable in practice as shown in Example \ref{ex:rep}.  
However, time and the interplay between entity levels, decay and activity durations require  a complex machinery to realize the expected behaviour into high-level Petri nets. Nevertheless, the advantage is that this realization is completely transparent to the final user.

In order to properly describe decay and the semantics of activities we need to record three kinds of information:
\begin{enumerate}
 \item (decay) since how long an entity is at the current level since the last update,
 \item (activators) since how long an entity is available at a level less or equal to the current one (recall that levels are inclusive),  
 \item (inhibitors) since how long a level greater than the current one has been left.
\end{enumerate}
Notice that duration relative to Items 1 and 2 for the current level can be different: \ie  if an entity is continuously sustained
by some activities it remains available in the system at the current
level for a period that may be longer than the corresponding decay
time. Thus  Item 1 reports the duration since the \emph{last} update (\ie when the activity concerning the entity has been performed) and Item 2 is the duration since the \emph{first} appearance of the entity at that level.

For this reason, we model each entity $\entity \in \entities$ by a single Petri net place $p_{\entity}$ carrying  tuples
$\tuple{\bool_\entity, \refr_\entity, \birth_\entity}$ as tokens where
$\bool_\entity$ is  the current level of $\entity$; 
$\refr_\entity$ is a counter storing the duration spent in the current level since the last update (Item 1): $0 \leq \refr_\entity \leq \life_\entity(\bool_\entity)$; 
and $\birth_\entity$ is a tuple of counters with $\setlev_\entity$ fields (one for
each level).  Each counter $\birth_{\entity}[i]$ contains a
duration interpreted as follows:
$$
\birth_{\entity}[i] = 
\begin{cases}
\begin{minipage}{7cm}
  since how long $\entity$ has reached level  $i$ (Item 2)
\end{minipage}
 &  \text{for }  0 \leq i \leq \lev_{\entity}\\
\\
\begin{minipage}{7cm}
   since how long $\entity$ has left level $i$ (Item 3)\\(or $0$ if $\entity$ has not yet reached $i$).
\end{minipage} 
& \text{for } \lev_{\entity} < i \leq \setlev_{\entity}-1              
\end{cases}
$$

Each place is initialised to $\tuple{\level{\entity}, 0,
  0^{\setlev_\entity} }$ where $\level{\entity}$ is the given initial
level of $\entity$ and $0^{\setlev_\entity}$ denotes vector $\birth_\entity$ of counters
uniformly initialised with zero.
 
It is worth observing that once a duration in $\birth$ has reached the maximum duration for all activities ($D = \max\{\Delta_\rho \mid \rho \in \Syn \cup \Reac \}$) it is useless to increment it, as for all duration greater than $D$, the guards of all activities involving the entity are satisfied.

\begin{newnotation}
      $\birth\sub{x}{[l \mydots l + n ]}$ is the systematic update to $x$ of
      values in counters $\birth[l] \mydots \birth[l+n]$. 
      We denote by $\inc{D}{\birth}$ the increment by one of all values in $\birth$:
      \ie $\forall k \in [0 \mydots \setlev_\entity-1]$, $\birth[k] =
      \min(\birth[k]+1, D)$ with $D \in \nat$.
\end{newnotation}

Every \potential activity $\alpha \in \Reac$ is modelled with  a transition $t_{\alpha}$ 
and a special place $p_{\alpha}$ that  ensures that the same activity is not executed more than once every $\dur{\alpha}$ time units, 
complying to Phase Ph. \ref{item:reaction} in Definition \ref{def:andy}. Figure \ref{fig:reaction} details the scheme of the resulting Petri net and defines the arc annotations. 
Input and output arcs between the same place and transition with the same label (read arcs) are denoted with a double-pointed arrow with a single label.  We comment on the conditions (implementing requirements in Definition \ref{def:enabledactivity}) and results of firing of $t_{\alpha}$ corresponding to a \potential activity $\alpha= \activ{A}{I}{\dur{}}{R}$. We have:
\begin{itemize}
\item since each activity must wait at least $\dur{}$ before being enabled again, place $p_{\alpha}$ retains whether the activity can be performed $w_{\alpha}\geq\dur{}$ and after firing the token is set to  $0$, ($w'_{\alpha}=0$).
 \item each activator $a \in A$ has to be present at least at  level $\level{a}$ for at least $\dur{}$ time units, this is expressed by  guard $\lev_a \geq \level{a} \wedge \birth[\level{a}] \geq \dur{} $; 
 \item each inhibitor $i\in I$ has not to exceed  level $\level{i}$ for at least $\dur{}$ time units,  this is guaranteed by guard 
$\lev_i < \level{i} \wedge \birth[\level{i}] \geq \dur{}$;
\item the updates on a place $p_r$ of a result entity $r$ at level $\lev_r$ containing token  $\tuple{\lev_r, \refr_r, \birth_r}$,  (\ie $(r,\pm n) \in R$ with $n\in \nat$) 
are performed\footnote{For the sake of simplicity, in the explanation we consider updates that do not exceed allowed boundaries, these cases are handled as expected in Definition \ref{def:rs}.} as follows: 
\begin{description}
 \item[Increase $+n$:]  
   $\tuple{\lev_r + n, 0, \birth_r\sub{0}{[\lev_r + 1\mydots \lev_r + n]}}$.
  $\lev_r$ becomes $\lev_r +n$, $\refr_r$ is reset to 0 restarting the counter  for level $\lev_r +n$ and 
$\birth_r$ is updated to record the change to the new level by resetting the counters $\birth_r[i]$ for all levels $i$ from $\lev_r + 1$ to $\lev_r + n$.
 
 \item[Maintenance $0$:]  $\tuple{\lev_r, 0, \birth_r}$: the current level is recharged resetting $\refr_r$ to 0.

  \item[Decrease $-n$:]  
	    $\tuple{\lev_r - n, 0, \birth_r\sub{0}{[\lev_r-n+1 \mydots \lev_r] }}$. This case is symmetric to the increase one, except for the treatment of $\birth_r$ that is updated by resetting  counters   $\birth_r[i]$
for all levels $i$  from  $\lev_r - n + 1$ to $\lev_r$.
\end{description}

\end{itemize}

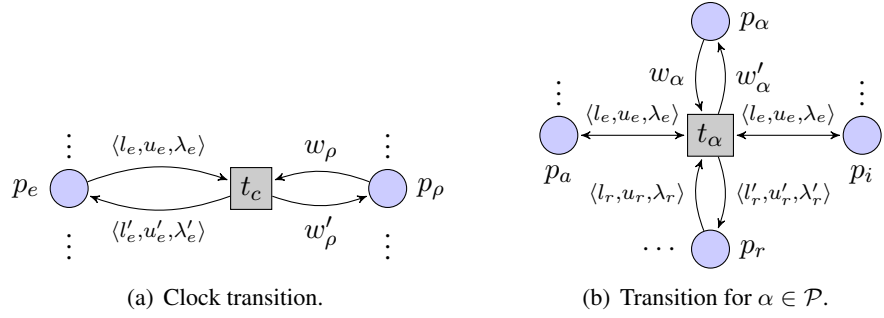
\begin{figure}[t]
\centering
\subfigure[Clock transition. \label{fig:clock}]{
\begin{tikzpicture}[xscale=1.2, yscale=1.2, node distance=1.5cm,>=stealth',bend angle=20,auto]

  \node [place, label = right: $p_{\rho}$, label = above: $\vdots$,label = below: $\vdots$] at (3.5,0) (Pc){};
  \node [place, label=left:$p_{\entity}$,label = above: $\vdots$,label = below: $\vdots$]at (0,0)  (R){};

  \node [transition] (tC) at (2,0) {$t_{c}$}
   edge [pre,bend left]    node[above] {$w_{\rho}$}    (Pc)
   edge [post,bend right]   node[below]{$w'_{\rho}$}    (Pc)
   edge [pre,bend right]    node[above] {${\scriptstyle \tuple{\bool_{\entity},\refr_{\entity},\birth_{\entity}}}$}    (R)
   edge [post,bend left]    node[auto] {${\scriptstyle \tuple{\bool_{\entity}',\refr_{\entity}',\birth_{\entity}'}}$}    (R);
  
\end{tikzpicture}
} \quad \quad 
\subfigure[Transition for $\alpha\in\Reac$. \label{fig:reaction}]{
\begin{tikzpicture}[node distance=1.5cm,>=stealth',bend angle=20,auto]

 \node [place, label = right: $p_r$,label = left: $\cdots$] at (0,-1.5) (P){};
 \node [place, label = below: $p_a$,label = above: $\vdots$] at (-2,0) (R){};
 \node [place, label = below: $p_i$,label = above: $\vdots$] at (2,0) (I){};
 \node [place, label = right: $p_{\alpha}$] at (0,1.5) (qw){};
 
 \node [transition] at (0,0) (tr) {$t_{\alpha}$}
   edge [pre and post]    node[above] {${\scriptstyle \tuple{\bool_{\entity},\refr_{\entity},\birth_{\entity}}}$}   (R)
   edge [pre and post]    node[above] {${\scriptstyle \tuple{\bool_{\entity},\refr_{\entity},\birth_{\entity}}}$}     (I)
   edge [pre, bend right]    node[left] {${\scriptstyle \tuple{\bool_r,\refr_r,\birth_r}}$}   (P)
   edge [post, bend left]   node[right] {${\scriptstyle \tuple{\lev'_r,\refr'_r,\birth_r'}}$}     (P)
   edge [pre, bend left]    node[left] {$w_{\alpha}$}   (qw)
   edge [post, bend right]   node[right] {$w'_{\alpha}$}     (qw);

\end{tikzpicture}
}
\caption{Scheme of Petri net modelling of \andy.}
\end{figure}

Finally, in order to cope with time aspects present in \andy, we introduce a tick transition $t_c$ (Figure \ref{fig:clock}) that represents the time progression.  Its effect is the one described by Phase Ph. \ref{item:clock} in Definition \ref{def:andy}: it increments  counters of entities,  takes care of the decay and simultaneously performs all enabled \obligatory activities. More precisely: 
      \begin{description}
      \item [Time: ] All counters in entities tuples are incremented by
            one: we pass from $\tuple{\lev, \refr, \birth}$ to
            $\tuple{\lev, \refr+1, \inc{D}{\birth}}$.  
            Counter $\refr$ remains unchanged for levels of unbounded
            duration.
      \item [Decay: ] An entity may stay at level $\lev$ for
            $\life(\lev)$ time units (\ie $\refr \leq \life(\lev)$).
            Decay happens as soon as the interval $\life(\lev)$ is
            elapsed and is obtained by decreasing the level by one, till
            reaching level zero.  
            We pass from
            $\tuple{\lev, \refr, \birth}$ to $\tuple{\lev - 1, 0,
              \inc{D}{\birth}\sub{0}{\lev}}$.
      \item [\Obligatory activities: ] All enabled \obligatory activities $\beta$ are performed simultaneously. The results on involved entities are collected and summed up. The derived update works as decribed above for \potential activities. 

      \end{description}

      Notice that transition $t_c$ is always enabled and Algorithm $\mathit{Calc}$ (cf., Algorithm \ref{fig:calc}) takes care of implementing the proper updates. The Algorithm is divided into  4 successive 
loops: the first one increments the counters in $p_{\rho}$ ($w'_{\rho} = w_{\rho}+1$) for all activities 
$\rho \in \Syn \cup \Reac$, the second one 
realizes decay for all entities, the third one checks whether a \obligatory activity $\beta$ is enabled in which case the 
involved entities are updated and the corresponding $w'_{\beta}$ is set to 0. The last loop combines the effect of decay and \obligatory activities and updates counters in $\birth$ accordingly.

Next definition puts together all the elements described so far:



\begin{definition}\label{def:rs}
 Given an  \andy network $(\entities, \Syn, \Reac)$ with initial 
state $(\entity, \level{\entity})$ for each $\entity \in \entities$, 
the high-level Petri net representation is defined as 
tuple $(P, T, F, L, M_0)$ where  $\lev$, $\lev'$, $\refr$, $\refr'$, $\birth$, 
$\birth'$, $w$, $w'$ range over variables and:
\begin{itemize} 
 \item $P = \{p_\entity \mid  \entity\in{\entities}\} \cup \{p_{\rho} \mid  \rho \in \Syn \cup \Reac \}; $
 \item $T= \{t_c \} \cup \{t_{\alpha} \mid \alpha \in \Reac \}$;

 \item $F= \{(p,t_c),(t_c, p) \mid p \in P \} \quad  \cup$ \\
 \hspace*{0.63cm} $ \{ (p_\entity,t_{\alpha}),(t_{\alpha}, p_\entity), (p_{\alpha},t_{\alpha}),
		(t_{\alpha}, p_{\alpha}) \mid \alpha \in \Reac, 
		\entity \in A_{\alpha} \cup I_{\alpha} \cup R_{\alpha} \}$

 \item Labels for places in $P$: 
\[
\begin{array}{lr}
  L(p_{\rho})  = [0 \mydots D] & \text{ for each } \rho \in \Syn \cup \Reac  \\
L(p_\entity) = [0\mydots\setlev_\entity-1] \times [0 \mydots d] \times  [0 \mydots D]^{\setlev_\entity}  
& \text{ for each } \entity\in\entities 
\end{array}
\]
with $d = \max\{\life_{\entity}(i)\mid i \in [0 \mydots \setlev_\entity-1]\}$ and $D = \max\{\dur{\rho}\mid \rho \in \Syn \cup \Reac\}$
\item Labels for arcs in $F$:
\[
\begin{array}{llr}
 L((p_{\rho}, t_c)) = w_{\rho} & L((t_c, p_{\rho} )) = w'_{\rho} & \text{ for each } \rho \in \Syn \cup \Reac\\
    L((p_\entity, t_c)) = \tuple{\bool_\entity, \refr_\entity, \birth_\entity} 
		& L((t_c, p_\entity)) = \tuple{\bool_\entity', \refr_\entity', \birth_\entity'} 
		& \text{ for each } \entity\in\entities
\end{array}
\]
For each \potential activity $\alpha\in\Reac$ and $\entity \in A_{\alpha} \cup I_{\alpha} \cup R_{\alpha}$:
\[
\begin{array}{llr}
      L((p_\entity, t_{\alpha})) = \tuple{\bool_\entity, \refr_\entity, \birth_\entity} \quad 
		& L((t_{\alpha}, p_\entity)) = \begin{cases}
                           \tuple{\bool_\entity, \refr_\entity, \birth_\entity} & \text{if } \entity\notin R_{\alpha}\\
                           \tuple{\bool_\entity', \refr_\entity', \birth_\entity'} & \text{otherwise}
                          \end{cases}\\

     L((p_{\alpha}, t_{\alpha})) = w_{\alpha} & L((t_{\alpha}, p_{\alpha} )) = w'_{\alpha}
\end{array}
\]

\begin{algorithm}[ht]
\caption{Calc’s algorithm}\label{fig:calc}
\begin{algorithmic}[0]
\Function{Calc}{}
 \ForAll{$\rho \in \Syn \cup \Reac$}
 \quad  $w_{\rho}'':=\max(w_{\rho}+1,D)$
 \EndFor
\ForAll{$\entity\in\entities$}
  \State $\lev_\entity'':=\lev_\entity $ \qquad $\refr_\entity'':=\refr_\entity$ \qquad $\birth_\entity'':=\inc{D}{\birth_\entity}$
  \If{ $\life(\lev_\entity) \neq \omega$ } \quad  $\refr_\entity'':= \refr_\entity + 1$
   \EndIf
\If {$ \refr_\entity > \life(\lev_\entity)$} \quad  $\lev_\entity'':= \max(0, \lev_\entity -1)$ \qquad $\refr_\entity'':= 0$
\EndIf
\EndFor

\ForAll{$\beta \in \Syn $}
 \If{
 $\forall {(a, \level{a}) \in A_{\beta}}(\lev_a \geq \level{a} ~ and~   \birth[\level{a}] \geq \dur{\beta} ) \quad   ~and$ \\
 \hspace{1.45cm}$\forall {(i, \level{i}) \in I_{\beta}}(\lev_i < \level{i} ~ and~  \birth[\level{i}] \geq \dur{\beta} )  \quad and \quad 
  (w_{\beta} \geq \dur{\beta}  )$ \ }
%
%
 \State $w''_{\beta} := 0$
   \ForAll{$ (r,v)\in R_{\beta}$} 
   \quad  $\lev_r'':= \lev_r'' +v $
   \qquad  $\refr_r'':= 0$
   \EndFor
\EndIf
\EndFor

\ForAll{$ \entity \in \entities$} 
 \State $\lev_e'':= \max(0, \min(\lev_e'', \setlev_e-1 ))$ 
 
 \If{$\lev_\entity'' < \lev_\entity$}
   \quad $  \birth_\entity':=\birth''_\entity\sub{0}{[\lev_\entity''+1\mydots \lev_\entity]}$
 \EndIf
 \If{$\lev_\entity'' > \lev_\entity$}
  \quad $\birth_\entity''=\birth_\entity''\sub{0}{[ \lev_\entity+1 \mydots \lev_\entity'']}$
 \EndIf
\EndFor
\State \textbf{return} $\bigwedge_{\rho \in \Syn \cup \Reac} (w'_{\rho}=w''_{\rho}) \wedge \bigwedge_{\entity\in\entities} (\tuple{\bool_\entity', \refr_\entity', \birth_\entity'} =\tuple{\bool_\entity'', \refr_\entity'', \birth_\entity''} )$ 
\EndFunction
\end{algorithmic}

\end{algorithm}

\item Labels for transitions in $T$:
 \[
 L(t_c) = \mathit{Calc}
 \]
where $\mathit{Calc}$ is the  guard computed 
by Algorithm  \ref{fig:calc}; when evaluated for a given valuation of net variables of input arcs of $t_c$, it produces the corresponding  updates of net variables of output arcs. 
 
\[
\begin{array}{lcl}
L(t_{\alpha}) \! &=&  w_{\alpha} \geq \dur{\alpha}  \wedge  w'_{\alpha} = 0 \quad \wedge \\
&& \bigwedge_{(a, \level{a}) \in A_{\alpha}}(\lev_a \geq \level{a} \wedge \birth[\level{a}] \geq \dur{\alpha} )  \quad \wedge\\
&&  \bigwedge_{(i, \level{i}) \in I_{\alpha}}(\lev_i < \level{i} \wedge \birth[\level{i}] \geq \dur{\alpha} )  \quad  \wedge \\

&& \bigwedge_{(r,+n) \in R_{\alpha}} C_{r+} \wedge \bigwedge_{(r,0) \in E_{\alpha}}C_{r0} \wedge  \bigwedge_{(r,-n) \in R_{\alpha}}C_{r-},\\  
\end{array}
\]
for each \potential activity $\alpha\in\Reac$ with
\[
\begin{array}{ll}
C_{r+} \!&:\!  \tuple{\lev'_r, \refr'_r, \birth'_r} = \tuple{\min(\lev_r +n, \setlev_r-1), 0, \birth_r\sub{0}{[\lev_r +1 \mydots \min(\lev_r +n, \setlev_r-1 ]}} \\
C_{r0} \!&: \! \tuple{\lev'_r, \refr'_r, \birth'_r} = \tuple{\lev_r, 0, \birth_r} \\

C_{r-} \!&: \! \tuple{\lev'_r, \refr'_r, \birth'_r} = \tuple{\max(0,\lev_r-n), 0, \birth_r\sub{0}{[\max(0,\lev_r-n)+1 \mydots \lev_r]}}.
\end{array}
\]

\item The initial marking is $M_0(p_{\entity})= \tuple{\level{s},0,0^{\setlev_\entity}}$ for each  $\entity\in \entities$ and $M_0(p_{\rho})= 0$ for 
$\rho\in \Syn \cup \Reac$.
\end{itemize}
\end{definition}

\begin{figure}[ht]
 \centering
 \includegraphics[width=0.65\textwidth]{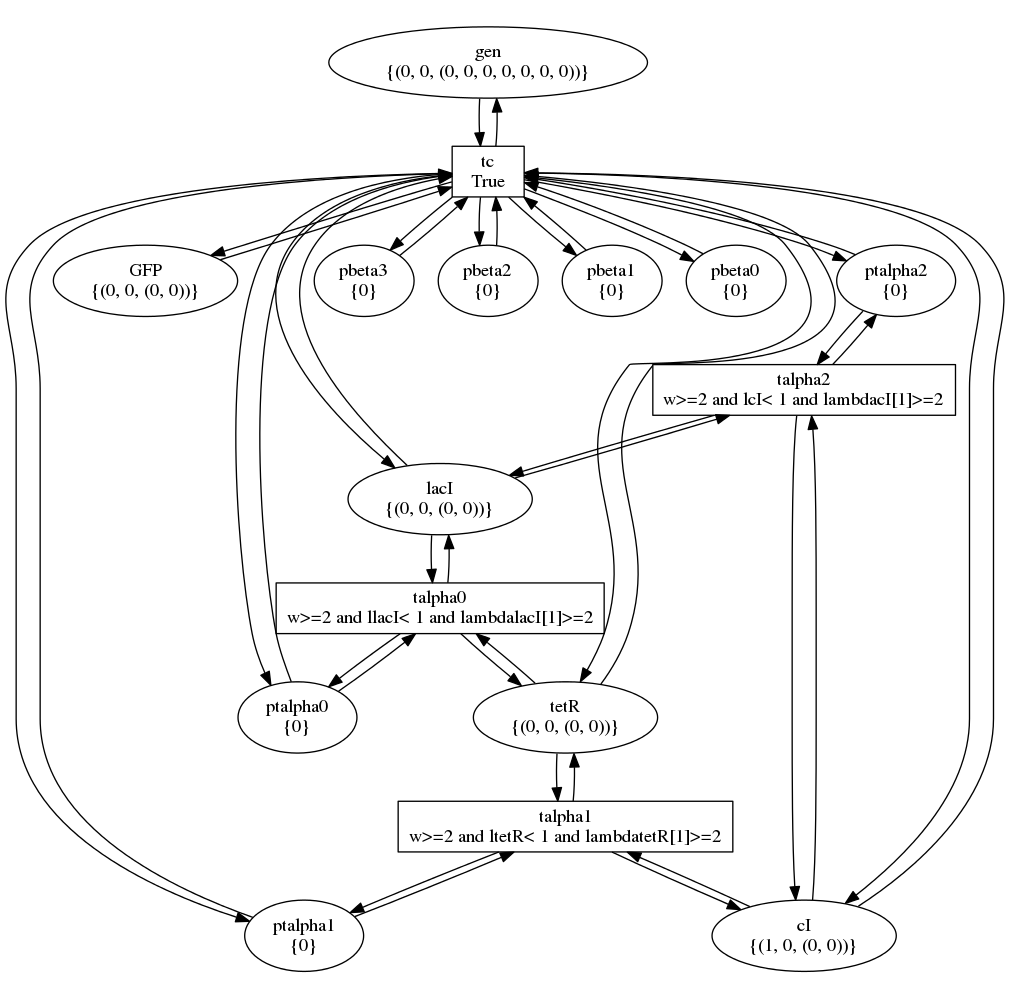}
\caption{The general shape of the \andy network for the Repressilator from Example  \ref{ex:rep} (arcs annotations are omitted).}\label{fig:repress}
\end{figure}

\begin{example}\label{ex:rep-petrinet}

Figure \ref{fig:repress} gives the simplified Petri net representation of the \andy system introduced in Example \ref{ex:rep},
obtained via the Snakes library.

\end{example}

\paragraph{\bf Expressiveness.}
The first result that we show for \andy networks concerns the finiteness of the state space. 
Actually, even if the number of species and levels is finite, as a state of the network includes an information concerning time (that could be unbounded, if one takes
a naive representation of dates), this could lead to an infinite state space. 
Nonetheless in our case, when the local counters have reached the maximal duration for activities $D$, 
no new behaviours can be triggered by the system. 
This immediately provides a way of abstracting without loosing precision and thus getting a symbolic representation. 

\begin{proposition}
The state space of  an \andy network $N$ is finite.
\end{proposition}
\begin{proof}
 Immediate as the type of each place is finite and each place can have at most one token. 
\end{proof}

An alternative formalization of \andy networks can be given using a timed automata model \cite{DBLP:journals/tcs/AlurD94}.  
 In Appendix \ref{app:timed} we provide an encoding that is compositional with respect to entities and \potential activities. This encoding is  exponential because of the nature of \obligatory activities that have to be treated simultaneously and dynamically.  We think that this encoding, although equivalent, is less elegant as it requires a lower level representation of the system. Indeed it corresponds to a sort of unfolding of the Petri net given in Definition \ref{def:rs}.

\section{Application to toxicology}\label{sec:toxic}
\andy and in particular the introduction of time constraints in conjunction with expression levels well adapts to study and explain toxic behaviours in biological systems. Here we explore the notion of toxicity and provide a methodology to analyse it in \andy.

The main approach used in toxicogenomics employs empirical analysis like in the identification of molecular biomarkers, \ie indicators of disease or toxicity in the form of specific gene expression patterns~\cite{DeCristofaro2008}.
Clearly, biomarkers remain observational indicators linking genes related measures to toxic states. In this proposal, we complement these empirical methods with a computational method  that aims at discovering the molecular mechanisms of toxicity.
This  way, instead of studying the  phenomenology of the toxic impacts, we focus on the processes triggering adverse effects on organisms.
Usually, the toxicity process  is defined as a sequence of physiological events that causes the abnormal behaviour of a living organism  with respect to its healthy state.

Healthy physiological states generally correspond to homoeostasis, namely a process that maintains a dynamic stability of internal conditions against changes in the external environment.
Hence, we will consider toxicity outcomes as deregulation of homoeostasis processes, namely deviation of some intrinsic referential equilibrium  of the system. 
Biological processes are usually given in terms of pathways which are  causal chains of the responses to stimuli, this way the deregulation of homoeostasis  appears as the unexpected activation or inhibition of existing pathways. Moreover,
in the context of toxicogenomics it is crucial to take into account at least two other parameters: the exposure time  and the thresholds dosage delimiting the ranges of safe and hazardous effects.  
Considering the qualitative nature of \andy models, we believe that the potential contribution lies in the analysis of the etiology of toxic processes. From this perspective, toxicity properties can be classified into two categories:
\begin{enumerate}
     \item \label{symptom} properties characterising pathological states or more generally associated to symptoms  
and
    \item \label{pathway} properties characterising sequences of states (traces), \eg non viable behaviours.
\end{enumerate}
The former class of properties basically leads to check the reachability of some states, while the latter may be used to unravel sequences of events leading to toxic outcomes. 

\andy allows us to describe such scenarios and the Petri net representation of a reaction network with the associated marking graph may be used to detect and predict toxic behaviours related to the dynamics of biological networks, both from the point of view of the dose (\eg some species exceed a pathological level) and from the point of view of the exposure (\eg some species persist in time beyond some maximal exposure time). Remark that defining toxicity by characterising properties on sequences of events is more general than just characterising pathological states. It addresses the needs to characterise multifactorial scenario (\eg periodical behaviours relating possibly several species) and to discover complex etiology among reactants. 

\paragraph{\bf From healthy states to healthy behaviours.}
In \andy a state is a marking in the marking graph, it collects for each species its current level and the corresponding $\lambda$ vector. The first step is to observe that some states can be naturally characterised as dangerous, they highlight particular symptoms (e.g. the dosage of a poison has become relevant). Let \danger be the set of such states. Symmetrically to dangerous states, there are some states that characterise a ``safe'' condition (e.g. the organism is functioning normally and does not present hazardous symptoms), denoted \healthy. Notice that there are states that are not assigned to any of the two classes, we group them into set  \remain.
Characterising harmful states corresponds only to address toxicity problems as in Item \ref{symptom}  above, more complex properties (Item \ref{pathway}) need to refer to (sets of) traces in the marking graph. We therefore classify possible behaviours starting from an initial state into four toxicity scenarios: 

\begin{enumerate}
\item those that lead to  states in \danger, 
\item those that irremediably quit \healthy and never reach it again: the tail of the trace visits only states in \danger or \remain (\eg when a gene mutation knock out some  pathway);

\item those that quit \healthy for too long in terms of ticks (even if they may reach it again), the regulation system allows, in principle, to reestablish the homeostatic equilibrium but the process takes too long, compromising the viability of the organism (\eg DNA damage caused by hydrogen peroxide in Escherichia Coli can be recovered to some extent by DNA repair enzymes but the perturbation may lead to cell apoptosis if to important or too long \cite{Imlay03061988});


\item those that quit \healthy an unbounded number of times (systemic metabolic dysfunction, slow poisoning) 

\end{enumerate}

All these properties can be given in terms of temporal logic formulae. As mentioned above,  the first item corresponds to reachability analysis thus instead of logic formulae ad hoc algorithms can be used as well. The second item can be dealt with formulae like the following in CTL: EF EG $ (s \notin \healthy)$, which means that eventually  (EF)  all visited states (EG) will not belong to set \healthy. The formulation of remaining items depends on the specific property of the system under consideration, for instance taking into account a precise number of ticks or a concrete sequence of states. 

\paragraph{\bf From healthy behaviours to healthy states.}
Notice that sometimes it is not sufficient or possible to describe the healthy condition of an organism (\ie the ability of maintaining its capabilities) only in terms of the sets \healthy and \danger. Indeed, for some organisms their viability conditions are characterised by complex behaviours such as specific traces or oscillations. We, thus, define toxic (and symmetrically viable) behaviours $\mathcal{T}_{toxic}$ (resp. $\mathcal{T}_{viable}$) by giving the property characterising them (or simply by enumerating the traces).  The properties can be described by specific temporal logic formulae (that can be rather complex employing also recursive operators). For instance, we can consider viable all traces that infinitely repeat an ordered sequence of three states  as in our example of the repressilator (\eg only lacL ON, then only tetR ON, and finally only cI ON)  and where the period of oscillation is six ticks. 
This way toxicity can be checked by testing whether a trace belongs or not to $\mathcal{T}_{toxic}$.  
For instance, suppose that we want to describe the periodic activation of three genes $a, b$ and $c$, we could use this $\mu$ calculus formula to describe $\mathcal{T}_{viable}$:
$$\nu.X (a \wedge \langle-\rangle\langle-\rangle\langle-\rangle tt \wedge [-](b \wedge [-](c \wedge [-]X )))$$
which means we are in a state with $a$ activated, followed by three transitions, for each first transition we should have $b$ activated and then $c$ and  so on recursively. 
In our case, as the state space is finite, any \andy network may be simulated by a B\"uchi automaton and the questions above may correspond either to check word acceptance or language inclusion, which are both decidable questions in our case.

Moreover once we have given the specification of viable and toxic behaviour, we may infer several interesting sets of states:
\begin{itemize}
\item  states starting from which all possible futures are behaviours in $\mathcal{T}_{viable/toxic}$ 
\item states that have at least one possible future in $\mathcal{T}_{viable/toxic}$,
\item states that are included only in $\mathcal{T}_{viable/toxic}$.
\end{itemize}

\begin{example}
We generalise here the example of the repressilator given above: we take three entities 
$\{A, B, C\}$ each with decay duration 2 that are inhibited in turn by potential activities:
\[
\begin{array}{l}
 \alpha_1: \activ{\emptyset}{(A,1)}{2}{(B,+1)} \qquad
 \alpha_2: \activ{\emptyset}{(B,1)}{2}{(C,+1)} \qquad 
 \alpha_3: \activ{\emptyset}{(C,1)}{2}{(A,+1)}  
\end{array}
\]
we consider a new perturbed variant of the system with an additional entity $D$ with only one level (\ie the entity is permanent) such that $\alpha_1$ and $\alpha_3$ remains unchanged and $\alpha_2$ is modified in  $$\alpha_2': \activ{(D,1)}{(B,1)}{4}{(C,+1)}$$ 

 \begin{figure}[ht]
 \centering
 \begin{tikzpicture}[>=latex',xscale=.4, yscale=.5]
 
 \draw[draw=none,fill=blue!10] (-1,1) rectangle (22,2) ; 
 \draw[draw=none,fill=blue!10] (-1,3) rectangle (22,4) ; 
 
 \node at (-1,1) (p0) {};
 \node[left] at (p0.west) {$0$};
 \node at (-1,2) (p1) {};
 \node[left] at (p1.west) {$1$};
 \node at (-1,3) (c0) {};
 \node[left] at (c0.west) {$0$};
 \node at (-1,4) (c1) {};
 \node[left] at (c1.west) {$1$};
 
  \node[label=left:{$A'=\left[ \phantom{\begin{array}{l} a\\ a \end{array}} \right.$}] at (-1,3.5) (c) {};
 \node[label=left:{$A=\left[ \phantom{\begin{array}{l} a\\ a \end{array}} \right.$}] at (-1,1.5) (p) {};
 
 
 \node at (-2,0) (t0) {};
 \node[rectangle,draw,fill=blue,scale=.3,label=below:0] at (0,0) (t1) {};
 \node[rectangle,draw,fill=blue,scale=.3,label=below:2] at (2,0) (t2) {};
 \node[rectangle,draw,fill=blue,scale=.3,label=below:4] at (4,0) (t3) {};
 \node[rectangle,draw,fill=blue,scale=.3,label=below:6] at (6,0) (t4) {};
 \node[rectangle,draw,fill=blue,scale=.3,label=below:8] at (8,0) (t5) {};
 \node[rectangle,draw,fill=blue,scale=.3,label=below:10] at (10,0) (t6) {};
 \node[rectangle,draw,fill=blue,scale=.3,label=below:11] at (12,0) (t7) {};
 \node[rectangle,draw,fill=blue,scale=.3,label=below:12] at (14,0) (t8) {};
 \node[rectangle,draw,fill=blue,scale=.3,label=below:13] at (16,0) (t9) {};
 \node[rectangle,draw,fill=blue,scale=.3,label=below:14] at (18,0) (t10) {};
 \node[rectangle,draw,fill=blue,scale=.3,label=below:16] at (20,0) (t11) {};
 
 \node[label=below right:$time$] at (22,0) (t12) {};
 
 \node[rectangle,draw,fill=blue,scale=.3] at (-1,1) (k1) {};
 \node[rectangle,draw,fill=blue,scale=.3] at (-1,2) (k2) {};
 \node[rectangle,draw,fill=blue,scale=.3] at (-1,3) (k3) {};
 \node[rectangle,draw,fill=blue,scale=.3] at (-1,4) (k4) {};
 
 \draw[->] (t0) -- (t12) ; 
 \draw[->] (-1,0) -- (-1,5.5) ; 
 
 \path[draw,thick] (0,2) --(2,2) --node[midway,above,sloped]{} (2,1) -- (6,1) --node[midway,above,sloped]{} (6,2) -- (8,2)  --node[midway,above,sloped]{} (8,1) -- (12,1) -- (12,2) -- (14,2) -- (14,1) -- (18,1) -- (18,2) -- (20,2) -- (20,1) -- (21,1);
 \node[right] at (21,1) {$\cdots$};
 \path[draw,thick] (0,4)--(2,4)  --node[midway,above,sloped]{} (2,3) -- (8,3)  --node[midway,above,sloped]{} (8,4) -- (10,4) -- (10,3) -- (16,3) -- (16,4) -- (18,4) -- (18,3) -- (21,3);
 \node[right] at (21,3) {$\cdots$};
 
 \end{tikzpicture}
 \caption{Entities levels evolving with time corresponding to the given scenario.}
 \label{fig:diagram}
 \end{figure}
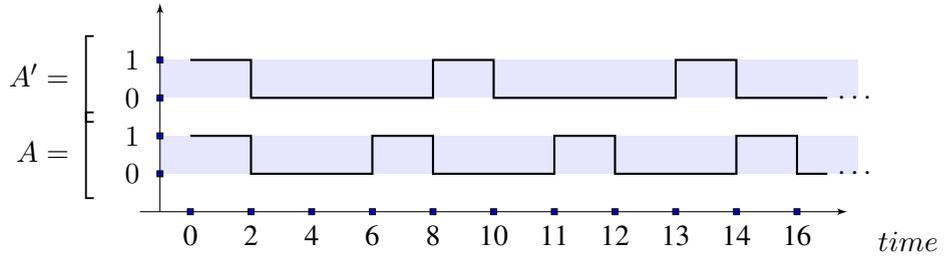

Figure \ref{fig:diagram} shows the periodic activation of entity $A$ (bottom line) each six time units while the second line shows the periodic behaviour of $A$ (top line, here denoted with $A'$ to avoid confusion) in the system perturbed with entity $D$.

We now define $\mathcal{T}_{viable}$ as all traces where entity $A$ is activated every six ticks, clearly with this definition of viability all   traces in the  perturbed system will be toxic.

\hfill$\diamond$
\end{example}

Notice that formulae describing properties can be rather long requiring, for instance, a complex combination of all possible cases and interleaving. To this aim, in the future we plan to design a language with dedicated primitives that abbreviate common use cases (\ie macro instructions).

\section{Case study: Blood glucose regulation}\label{sec:example}
Here we model \emph{glucose regulation} in human body (Figure \ref{fig:glucose}). This example shows that \andy model is not limited to genetic regulations but it may express other kinds of interactions. 
In the following, we are always referring to the process under normal circumstances in a healthy body.

Glucose regulation is a homeostatic process: \ie the rates of glucose in blood (\emph{glycemia}) must remain stable at what we call the equilibrium state. 
Glycemia is regulated by  two hormones: \emph{insulin} and  \emph{glucagon}. When glycemia rises (for instance as a result of the digestion of a meal), insulin  promotes the storing of glucose in muscles through the  glycogenesis process, thus decreasing the blood glucose levels. 
Conversely, when glycemia is critically low, glucagon  stimulates the process of glycogenolysis that increases the blood glucose level by transforming glycogen back into glucose.

 \begin{figure}[b]
\centering
\includegraphics[width=0.5\textwidth]{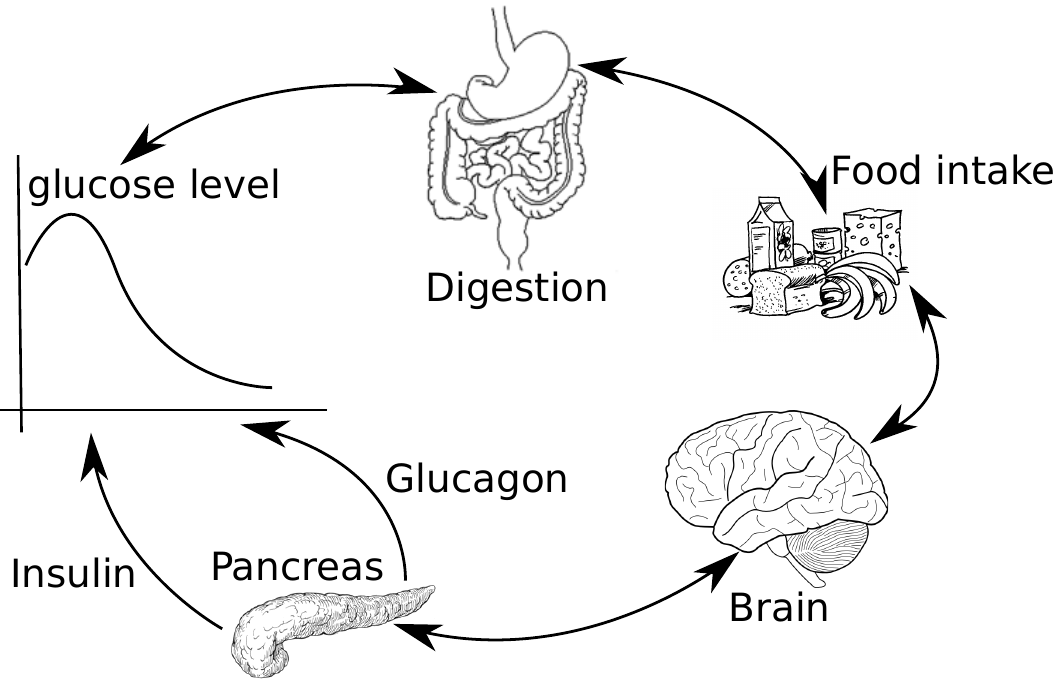} 
\caption{Glucose metabolism} \label{fig:glucose}
\end{figure}

We will focus on the assimilation of sweeteners: \ie sugars  or  artificial sweeteners such as aspartame. Whenever we eat something sweet either natural or artificial, the sweet sensation sends a signal to the brain (through \emph{neurotransmitters}) that in turns stimulates the production of insulin by pancreas. In the case of sugar, the digestion transforms food into nutrients (\ie glucose) that are absorbed by blood.
This way, sugar through digestion  increases glucose in blood giving the sensation of satiety. In case the income of glucose produces hyperglycemia,  the levels of glucose are promptly equilibrated by the intervention of insulin.
Unlike sugar, artificial sweeteners are not assimilated by the body, hence they do not increase the glucose levels in blood. Nevertheless  the insulin produced under the stimuli originated by the sweet sensation, although weak, can still cause the rate of glucose to drop engendering hypoglycemia. In response to that, the brain induces the  stimulus of \emph{hunger}. As a matter of fact this appears as an unwanted/toxic behaviour. Indeed the assimilation of food (even if it contains aspartame) should calm hunger and induce satiety not the opposite. 

This schema suggests that we should consider four  levels for glycemia: low, hunger, equilibrium and high. Likewise for insulin we assume three levels: inactive, low and high. All other actors involved in glucose regulation, have only two  levels (inactive or active).  
In this example, duration do not play a fundamental role, for the sake of simplicity we have set all complementary activities such as production of insulin and glucagon,  to take the same amount of time, the signal to the brain is the fastest, and the decay of  glycemia values are much longer than the digestion process. 

Thus the set of involved entities is $$\entities = \{ Sugar, Aspartame, Glycemia, Glucagon, Insulin\}$$ and their expression levels and  corresponding decays are:
$$ 
\begin{array}{|lcl|}
\hline
  \mbox{levels} & \hspace{2cm} & \mbox{duration} \\
	\hline
	\setlev_{sugar}=\{0,1\} && \life_{sugar}(1)=2 \\
	\hline
  \setlev_{aspartame}=\{0,1\} && \life_{aspartame}(1)=2 \\
	\hline
  \setlev_{glycemia}=\{0,1,2,3\} \quad && \life_{glycemia}(1)=8 \\
  && \life_{glycemia}(2)=8\\
  && \life_{glycemia}(3)=8\\
	\hline
  \setlev_{glucagon}=\{0,1\} && \life_{glucagon}(1)=3\\
	\hline
  \setlev_{insulin}=\{0,1,2\} && \life_{insulin}(1)=3\\
  && \life_{insulin}(2)=3\\
	\hline
\end{array}
$$
 The levels of glycemia are: 0 corresponding to low, 1 to hunger, 2 to equilibrium and 3 to high. Likewise for insulin we have 0 that corresponds to inactive, 1 to low and 2 to high. All levels for the other species are 0 for inactive and 1 for active.

The set of \potential activities $\Reac = \{\alpha_k=\activ{A_k}{I_k}{}{R_k} \mid k\in [1..9]\}$ for the glucose metabolism example  is:
$$\begin{array}{l l}
  \alpha_1: \ & \activ{(Sugar,1)}{\emptyset}{}{(Insulin,+1), (Glycemia,+1)} \\
 
	\alpha_2: & \activ{(Aspartame,1)}{ \emptyset }{}{(Insulin,+1)} \\
  
	\alpha_3 :& \activ{\emptyset}{(Glycemia,1)}{}{(Glucagon,+1)} \\

	\alpha_4 :& \activ{(Glycemia,3)}{ \emptyset}{}{ (Insulin,+1)} \\
  
	\alpha_5 :& \activ{(Insulin,2)}{ \emptyset}{}{(Glycemia,-1)} \\
 
	\alpha_6 :& \activ{(Insulin,1),(Glycemia,3)}{\emptyset}{}{(Glycemia,-1)} \\
 
	\alpha_7 :& \activ{(Insulin,1)}{  (Glycemia,2)}{}{(Glycemia,-1)} \\
  
	\alpha_{8}: & \activ{(Glucagon,1)} {\emptyset}{}{(Glycemia,+1)}  
   \end{array}
 $$
$\alpha_1$ and $\alpha_2$ represent the assimilation of Sugar and Aspartame, respectively: while Aspartame only increases the level of Insulin, Sugar also increases Glycemia.  $\alpha_3$ takes care of hypoglycemia, \ie a Glycemia level equal to 0 (obtained by using $(Glycemia, 1)$ as inhibitor) engenders the production of Glucagon. On the contrary, hyperglycemia causes the production of Insulin ($\alpha_4$). The presence of Insulin lowers Glycemia (activities $\alpha_5, \alpha_6, \alpha_7$). In particular Insulin level equal to 1 plays a role in the decrease of Glycemia only in case of hyperglycemia $\alpha_6$ or hypoglycemia $\alpha_7$, otherwise the signal is not strong enough and we need Insulin at level 2 to see the effect on Glycemia ($\alpha_5$). Last activity describes the role of Glucagon which if active increases the level of Glycemia.

For this example we consider an empty set of \obligatory activities. Observe now the behaviour of
$Glycemia$ in the following scenario:
$$
\begin{array}{ll}
 \text{initial state} & \tuple{3, 0}  \\
 \text{$8$ time units elapse, counter at level $3$ updates} & \tuple{3, 8} \\
 \text{one time unit elapses, $Glycemia$ decays} & \tuple{2, 0} \\
  \text{one time unit elapses, counter at level $2$ updates} & \tuple{2, 1}\\
  \text{activity $\alpha_5$ decreases $Glycemia$ level} & \tuple{1, 0} \\
  \text{$8$ time units elapse, counter at level $1$ updates} & \tuple{1, 8} \\
\text{one time unit elapses, $Glycemia$ decays} & \tuple{0, 0}\\
  \text{one time unit elapses, no effect since $\life_{glycemia}(0)=\omega$ \quad} & \tuple{0, 0}. 
\end{array}
$$

Figure \ref{fig:randyglucose} shows a simplified \andy network $(\entities, \Reac, \emptyset)$. The shown fragment focuses only on the activity schema linking inputs (\ie activators and inhibitors) to results. 
Each input arc is labeled with either  letter A or letter I denoting whether the input place is an activator or an inhibitor,  respectively. Likewise, each output arc is labeled with a + or a - to denote increase or decrease of product levels by 1.
For each activity transition $\alpha$, we have omitted place $q_{\alpha}$ and all arcs in the opposite direction. 
The numbers inside each transition refer to the corresponding activity.
Figure \ref{fig:zoomrandyglucose}, instead, shows a portion of the complete 
initially marked \andy network focusing only on activity $\alpha_7$. 

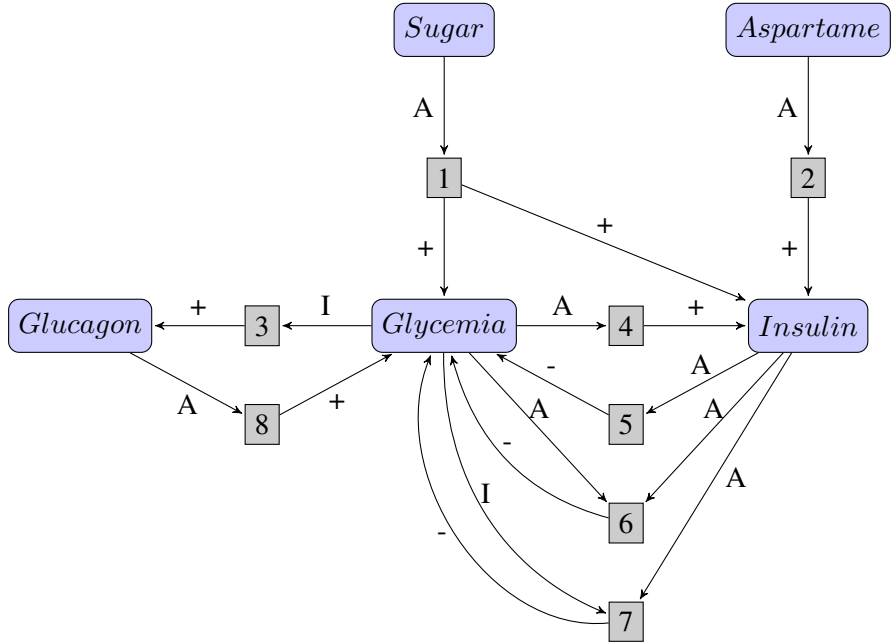
\begin{figure}[ht]
\centering
\begin{tikzpicture}[xscale=1.2, yscale=1.3, node distance=1.5cm,>=stealth',bend angle=45,auto]

\tikzstyle{place}=[rectangle,rounded corners,draw,fill=blue!20,minimum size=7mm]

 \node [place, label = above: ] at (2,7) (sugar){$Sugar$};
 \node [place, label = above: ] at (6,7) (asp){$Aspartame$};
 \node [place, label = right: ] at (2,4) (glucose){$Glycemia$};
 \node [place, label = right: ] at (-2,4) (gluc){$Glucagon$}; 
 \node [place, label = right: ] at (6,4) (ins){$Insulin$}; 

\node [transition, label={ right:}] at (2,5.5) (tr1) {1}
  edge [pre]    node[left] {A}    (sugar)
  edge [ post]    node[above] {+}   (ins)
  edge [ post]    node[left] {+}   (glucose);
  
\node [transition, label={ right:}] at (6,5.5) (tr2) {2}
  edge [pre]    node[left] {A}    (asp)
  edge [post]    node[left] {+}   (ins);

  \node [transition, label={ right:}] at (0,4) (tr3) {3}
  edge [pre]    node[above] {I}    (glucose)
  edge [ post]    node[above] {+}   (gluc);

\node [transition, label={right:}] at (4,4) (tr4) {4}
  edge [pre]    node[above] {A}    (glucose)
  edge [post]    node[above] {+}   (ins);

 \node [transition, label={right:}] at (4,3) (tr5) {5}
   edge [pre]    node[above] {A}    (ins)
   edge [post]    node[above] {-}   (glucose);
   
   \node [transition, label={ right:}] at (4,2) (tr6) {6}
   edge [pre]    node[above] {A}    (glucose)
   edge [pre]    node[above] {A}    (ins)
   edge [post, bend left=30]    node[above] {-}   (glucose);

\node [transition, label={ right:}] at (4,1) (tr7) {7}
  edge [pre, bend left=35]    node[above] {I}    (glucose)
  edge [pre]    node[right] {A}    (ins)

  edge [post, bend left =60]    node[left] {-}   (glucose);

 \node [transition, label={ right:}] at (0,3) (tr8) {8}
  edge [pre]    node[below] {A}    (gluc)
  edge [post]    node[below] {+}   (glucose);

\end{tikzpicture}
 \caption{Simplified \andy network of glucose metabolism. } \label{fig:randyglucose}
\end{figure}

\begin{figure}[ht]
\centering
\begin{tikzpicture}[node distance=1.5cm,>=stealth',bend angle=45,auto]

 \tikzstyle{place}=[circle,rounded corners, draw,fill=blue!30,minimum size=7mm]
 \tikzstyle{transition}=[rectangle,thick,draw=black!75,
 			 fill=black!20,minimum size=4mm]
 \tikzstyle{every token}=[font=\small]

 \node [place, label = above:$q_{glycemia}$ ] at (0,4) (gl){$(3,0)$}; 
 \node [place, label = below: $q_{insulin}$ ] at (0,0) (ins){$(0,0)$}; 
 \node [place, label = right:$q_{\alpha_7}$ ] at (6,2) (qw){1};
 
\node [transition, label={ left:$L(t_c)$} ] at (-3,2) (tr1) {$t_c$}
  edge [pre, bend left]    node[left] {\ $\quad \tuple{\lev_{g},\refr_{g}}$}    (gl)
  edge [post]    node[right] {$\tuple{\lev'_{g},\refr'_{g}}$}    (gl)
  edge [pre, bend right ]    node [left] { $\tuple{\lev_i,\refr_i}$ }    (ins)
  edge [post ]    node[right] {$\ \tuple{\lev'_i,\refr'_i}$ }    (ins);
  
\node [transition, label={ right:$L(t_{\alpha_7})$}] at (3,2) (tr8) {$t_{\alpha_7}$}
  edge [pre, bend left=30 ]    node[above] {\qquad \quad $\tuple{\lev_g,\refr_g}$}    (gl)
  edge [pre and post]    node[right] {$\ \tuple{\lev_i,\refr_i}$}    (ins)
  edge [post,bend right]    node[right] {
  $ \tuple{\lev'_{g},\refr'_{g}} $}   (gl)
    edge [pre, bend left]    node[below] {$w$}   (qw)
  edge [post,bend right]    node[above] {$0$}   (qw);
\end{tikzpicture}
 \caption{A portion of the \andy network of glucose metabolism with an initial marking. } \label{fig:zoomrandyglucose}
\end{figure}
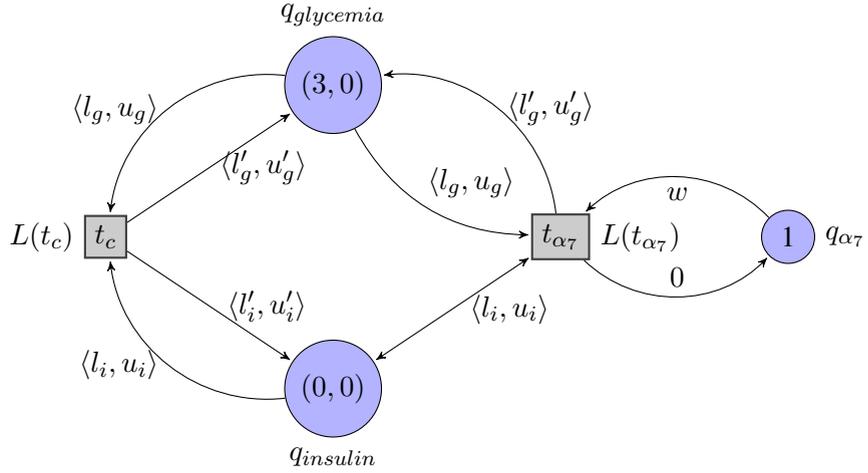

The obtained network allows to simulate and intuitively confirm expected behaviours. In addition  and as explained in previous section, we can perform exhaustive  checks and for instance automatically verify the following CTL properties:

\begin{description}
 \item[Symptoms:] Is it possible to have an anomalous decrease of glucose levels in blood (revealing hypoglycemia)? 
$$\esiste \finally (Glycemia, 0)$$
 \item [Causality:] Does assimilation of sweeteners cause hypoglycemia? 
$$\begin{array}{l}
\esiste \finally [((Sugar, 1) \vee (Aspartame, 1)) \wedge (Glycemia, 1)] \rightarrow  \all \finally (Glycemia,2)   
  \end{array}
 $$
\end{description}

For this formula we  highlight the different mode-of-action depending on the absorption of sugar or aspartame.

In the case of the assimilation of sugar, it induces an increase of the production of insulin and an augmentation of the blood glucose levels. Nonetheless the levels of insulin produced are not enough to cause the glycemia to drop and the formula is satisfied.
$$
\begin{array}{l}
(Sugar, 1), (Aspartame, 0), (Glycemia, 1), (Insulin, 0), (Glucagon, 0) \xrightarrow{\alpha_1}\\
(Sugar, 1), (Aspartame, 0), \mathbf{(Glycemia, 2)}, (Insulin, 1), (Glucagon, 0) 
\end{array}
 $$
 
In the case of the assimilation of aspartame but not sugar, it causes only an increase of insulin. Unfortunately, this increment is sufficient to induce a decrease of blood glucose levels thus contradicting the formula above. 
This illustrates the toxic behaviour caused by aspartame.
$$
\begin{array}{l}
 (Sugar, 0), (Aspartame, 1), (Glycemia, 1), (Insulin, 0), (Glucagon, 0) \xrightarrow{\alpha_2} \\ 
 (Sugar, 0), (Aspartame, 1), (Glycemia, 1), (Insulin, 1), (Glucagon, 0) \xrightarrow{\alpha_7} \\
 (Sugar, 0), (Aspartame, 0), \mathbf{(Glycemia, 0)}, (Insulin, 1), (Glucagon, 0)
\end{array}
 $$

The \andy network, its implementation in Snoopy and formulas for Charlie analyser as well as some results can be found in \cite{webpage}.

%
%

\section{Related works}\label{sec:related}

From a technical point of view, the closest related work is on reaction
systems \cite{DBLP:journals/ijfcs/BrijderEMR11} or their Petri net
representation \cite{koutnyreaction}. Although we use a similar
definition for activity, the semantics we have proposed is inherently
different: in \cite{DBLP:journals/ijfcs/BrijderEMR11} all enabled
reactions occur in one step while we consider two forms of activity evolution: simultaneous  for \obligatory activities and
interleaved for \potential ones.
Furthermore, we have introduced discrete abstract  levels that, to the
best of our knowledge, are not taken into account in reaction systems. 
Also, we have presented a more elaborated notion of time that governs
both entities and activities. 
In  \cite{DBLP:conf/birthday/BrijderER11} the authors consider an
extension of reaction systems with duration but it concerns only decay
and not  duration of activities. Furthermore the decay is  referred to
the number of steps and not to the discrete time progression as in our case.\looseness=-1  

From the Petri net modelling point of view, our representation of time is considerably different from the approaches traditionally used in time and timed Petri nets 
(\cite{DBLP:journals/tcs/CeroneM99} presents a survey with insightful
comparison of the different approaches).
The main difference lies on the fact that the  progression of time is implicit and  external to the system. By contrast, in our proposal we have assumed the presence of  an explicit way of incrementing duration (modelled by synchronised counters). 
This is also  different from the notion of timestamps introduced in  \cite{timestamp} that again refers to an implicit notion of time.
Indeed,  our approach is conceptually closer to Petri nets with causal time  \cite{DBLP:journals/entcs/ThanhKP02} for the presence of an explicit transition for time progression. Nevertheless, in our approach time cannot be suspended under the influence of the environment (as is the case in \cite{DBLP:journals/entcs/ThanhKP02}).\looseness=-1

Time is featured also in  \cite{DBLP:journals/procedia/CometBDDMC12} based on the Thomas framework and using  delays as parameters  to be discovered.
A similar approach  but based on timed automata is presented in  \cite{DBLP:conf/formats/BattSM07} and \cite{DBLP:conf/cmsb/SiebertB06}. 
 In all these proposals, delays are reminiscent of our durations for \obligatory activities but our semantics is more general as it includes the treatment of decay and \potential activities.
 
In a broader sense, our work could also be related to P-systems
\cite{DBLP:journals/ijfcs/PaunPRS11,DBLP:journals/fuin/KleijnK11} or the
$\kappa$-calculus \cite{DBLP:journals/tcs/DanosL04} that  describe the
evolution of cells through rules. 
Both these approaches are mainly oriented to simulation while we are
interested in  verification aspects. Moreover, always related to the
modelling in Petri nets but with a different aim, levels have been used
in qualitative approaches to address problems related to the
identification of steady states in genetic networks such as in
\cite{DBLP:journals/nc/ChaouiyaNRT11}. 
Nevertheless these contributions
abstract away from time related aspects that are instead central in our
proposal.   It is also interesting to notice that our formalism adds a new node in between timed and qualitative regions in the classification in   \cite{DBLP:conf/apn/HeinerG11}.

\section{Final remarks}\label{sec:concl}

We have introduced \andy 
to model
time-dependent activity-driven systems, which  consist of a set of entities
present in the environment at a given level. Entities can age as time
passes by and their level is governed by a set of \obligatory and \potential activities with duration. 
The \andy network is compact and provides an intelligible
representation of the described activities. Moreover, it is easily scalable as the size of the network grows linearly with
the number of added entities and activities and the rule-based architecture of guarantees a modular construction of the systems.
We have shown that, despite the ability of modelling timed (thus infinite) systems, \andy networks have finite state space. 

Moreover we have discussed how toxicity problems can be addressed using our formalism. To the 
best of our knowledge this is the first attempt to give a general methodology, a systemic approach of toxicity analysis rather than a particular technique to verify them. We have exemplified our approach to toxicity on the blood glucose regulation example.

As the semantics of \andy is given in terms of high-level Petri nets, \andy networks can be easily used as overlay of existing implemented tools.
%
%
We have prototyped a
tool to simulate and build the state space of an \andy network 
using the SNAKES
toolkit and Snoopy/Charlie. 
The first results \cite{webpage} are very positive,  \andy is particularly suited to describe biological systems, in particular regulatory networks and pathways whose formalisation is based on rules (activities).

 We plan to develop the \andy
approach along several directions. We want to extend the decay function to arbitrary maps in $\setlev$. The
problem is to generalise the notion of increase and decrease in a
relevant way for the application. More structure can also be given to
the notion of entity. Characterising an entity by several attributes
(instead of only one level) raises the question of the causal
relationships between these attributes and their update
strategy. Another aspect that will be addressed in the future is to
consider dynamic models where entities are created and deleted, as
 in~\cite{Giavitto2012}. 
Finally, it would certainly be interesting to find a systematic way for approximating stochastic rates into activities duration and decays.

%

 \paragraph{\bf Acknowledgements.}
 We thank the anonymous reviewers of PNSE and BioPPN 2014 for their comments on the very first version of this paper. We are grateful to M. Heiner,  F. Pommereau, Y.-S. Le Cornec and A. Finkelstein for useful suggestions and discussions. 

\bibliographystyle{fundam}
\bibliography{biblio}

\begin{thebibliography}{10}

\bibitem{DBLP:journals/ijfcs/BrijderEMR11}
Brijder R, Ehrenfeucht A, Main MG, Rozenberg G.
\newblock A Tour of reaction Systems.
\newblock Journal of Foundations of Computer Science. 2011;22(7):1499--1517.

\bibitem{Wang04formalverification}
Wang F.
\newblock Formal verification of timed systems: A survey and perspective.
\newblock In: Proceedings of the IEEE; 2004. p. 1--23.

\bibitem{alon2006introduction}
Alon U.
\newblock An introduction to systems biology: design principles of biological
  circuits.
\newblock CRC press; 2006.

\bibitem{DBLP:journals/nc/BaldanCMS10}
Baldan P, Cocco N, Marin A, Simeoni M.
\newblock Petri nets for modelling metabolic pathways: a survey.
\newblock Natural Computing. 2010;9(4):955--989.

\bibitem{thomas1973boolean}
Thomas R.
\newblock Boolean formalisation of genetic control circuits.
\newblock Journal of theoretical biology. 1973;42:565--583.

\bibitem{journals/tcsb/Cardelli05}
Cardelli L.
\newblock Abstract Machines of Systems Biology.
\newblock Transactions on Computational Systems Biology. 2005;3737:145--168.

\bibitem{giavitto04a}
Giavitto JL, Malcolm G, Michel O.
\newblock Rewriting systems and the modelling of biological systems.
\newblock Comparative and Functional Genomics. 2004 Feb;5:95--99.

\bibitem{Waters2004}
Waters MD, Fostel JM.
\newblock {Toxicogenomics and systems toxicology: aims and prospects.}
\newblock Nature reviews Genetics. 2004 Dec;5(12):936--48.

\bibitem{Foster2007}
Foster WR, Chen SJ, He A, Truong A, Bhaskaran V, Nelson DM, et~al.
\newblock {A retrospective analysis of toxicogenomics in the safety assessment
  of drug candidates.}
\newblock Toxicologic pathology. 2007 Aug;35(5):621--35.

\bibitem{synthetic}
Serrano L.
\newblock Synthetic biology: promises and challenges.
\newblock Molecular Systems Biology. 2007;3(158).

\bibitem{CDGHKFD14}
{Di Giusto} C, Klaudel H, Delaplace F.
\newblock Systemic approach for toxicity analysis.
\newblock In: Proceedings of the 5th International Workshop on Biological
  Processes {\&} Petri Nets. vol. 1159 of {CEUR} Workshop Proceedings.
  CEUR-WS.org; 2014. p. 30--44.

\bibitem{Pom-PNN-2008}
Pommereau F.
\newblock Quickly prototyping {Petri} nets tools with {SNAKES}.
\newblock Petri net newsletter. 2008 10;(10-2008):1--18.
\newblock \href{http://www.ibisc.univ-evry.fr/~fpommereau/SNAKES} {SNAKES is
  available here}.

\bibitem{DBLP:conf/apn/HeinerHLRS12}
Heiner M, Herajy M, Liu F, Rohr C, Schwarick M.
\newblock Snoopy - A Unifying Petri Net Tool.
\newblock In: Petri Nets. vol. 7347 of LNCS. Springer; 2012. p. 398--407.

\bibitem{HSW15}
Heiner M, Schwarick M, Wegener J.
\newblock {Charlie – an extensible Petri net analysis tool}.
\newblock In: Devillers R, Valmari A, editors. Proc. PETRI NETS 2015. vol. 9115
  of LNCS. Springer; 2015. p. 200--211.

\bibitem{webpage}
Additional material;.
\newblock Available from: \url{http://prova}.

\bibitem{Elowitz}
Elowitz M, Leibler S.
\newblock A Synthetic Oscillatory Network of Transcriptional Regulators.
\newblock Nature. 2000;403(6767):335--8.

\bibitem{DBLP:series/eatcs/Jensen92}
Jensen K.
\newblock Coloured Petri Nets - Basic Concepts, Analysis Methods and Practical
  Use - Volume 1.
\newblock EATCS Monographs on TCS. Springer; 1992.

\bibitem{DBLP:journals/tcs/AlurD94}
Alur R, Dill DL.
\newblock A Theory of Timed Automata.
\newblock TCS. 1994;126(2):183--235.

\bibitem{DeCristofaro2008}
{De Cristofaro} M, Daniels K.
\newblock Toxicogenomics in Biomarker Discovery.
\newblock In: Essential Concepts in Toxicogenomics. vol. 460 of Methods in
  Molecular Biology. Humana Press; 2008. p. 185--194.

\bibitem{Imlay03061988}
Imlay J, Linn S.
\newblock DNA damage and oxygen radical toxicity.
\newblock Science. 1988;240(4857):1302--1309.

\bibitem{koutnyreaction}
Kleijn J, Koutny M, Rozenberg G.
\newblock Modelling Reaction Systems with Petri Nets.
\newblock In: BioPPN-2011; 2011. p. 36--52.

\bibitem{DBLP:conf/birthday/BrijderER11}
Brijder R, Ehrenfeucht A, Rozenberg G.
\newblock Reaction Systems with Duration.
\newblock In: Computation, Cooperation, and Life. vol. 6610 of LNCS. Springer;
  2011. p. 191--202.

\bibitem{DBLP:journals/tcs/CeroneM99}
Cerone A, Maggiolo-Schettini A.
\newblock Time-Based Expressivity of Time Petri Nets for System Specification.
\newblock TCS. 1999;216(1-2):1--53.

\bibitem{timestamp}
Hanisch HM, Lautenbach K, Simon C, Thieme J.
\newblock Timestamp Petri Nets in Technical Applications.
\newblock In: WODES '98; 1998. p. 321--326.

\bibitem{DBLP:journals/entcs/ThanhKP02}
Thanh CB, Klaudel H, Pommereau F.
\newblock {P}etri nets with causal time for system verification.
\newblock ENTCS. 2002;68(5):85--100.

\bibitem{DBLP:journals/procedia/CometBDDMC12}
Comet J, Bernot G, Das A, Diener F, Massot C, Cessieux A.
\newblock Simplified Models for the Mammalian Circadian Clock.
\newblock In: CSBio 2012. vol.~11 of Procedia Computer Science. Elsevier; 2012.
  p. 127--138.

\bibitem{DBLP:conf/formats/BattSM07}
Batt G, Salah RB, Maler O.
\newblock On Timed Models of Gene Networks.
\newblock In: Raskin J, Thiagarajan PS, editors. {FORMATS} 2007. vol. 4763 of
  LNCS. Springer; 2007. p. 38--52.

\bibitem{DBLP:conf/cmsb/SiebertB06}
Siebert H, Bockmayr A.
\newblock Incorporating Time Delays into the Logical Analysis of Gene
  Regulatory Networks.
\newblock In: {CMSB} 2006,. vol. 4210 of LNCS. Springer; 2006. p. 169--183.

\bibitem{DBLP:journals/ijfcs/PaunPRS11}
Paun A, Paun M, Rodr\'{\i}guez-Pat{\'o}n A, Sidoroff M.
\newblock P Systems with proteins on Membranes: a Survey.
\newblock Int Journal of Foundations of Computer Science. 2011;22(1):39--53.

\bibitem{DBLP:journals/fuin/KleijnK11}
Kleijn J, Koutny M.
\newblock Membrane Systems with Qualitative Evolution Rules.
\newblock Fundam Inform. 2011;110(1-4):217--230.

\bibitem{DBLP:journals/tcs/DanosL04}
Danos V, Laneve C.
\newblock Formal molecular biology.
\newblock TCS. 2004;325(1):69--110.

\bibitem{DBLP:journals/nc/ChaouiyaNRT11}
Chaouiya C, Naldi A, Remy E, Thieffry D.
\newblock Petri net representation of multi-valued logical regulatory graphs.
\newblock Natural Computing. 2011;10(2):727--750.

\bibitem{DBLP:conf/apn/HeinerG11}
Heiner M, Gilbert D.
\newblock How Might Petri Nets Enhance Your Systems Biology Toolkit.
\newblock In: Petri Nets. vol. 6709 of LNCS. Springer; 2011. p. 17--37.

\bibitem{Giavitto2012}
Giavitto JL, Klaudel H, Pommereau F.
\newblock Integrated regulatory networks {(IRNs)}: {Spatially} organized
  biochemical modules.
\newblock TCS. 2012;431(1):219--234.

\end{thebibliography}

\appendix
\section{Encoding into Timed Automata} \label{app:timed}

In Section \ref{sec:reaction} we have hinted at the construction of a timed automaton coming directly from the marking graph of an  \andy network. Here we present another encoding that is syntax-driven which is not straightforward because of the management
of \obligatory rules. Unfortunately, even if the obtained automata is smaller, we cannot avoid it to be exponential at least in the number of activities.
We conclude this section showing that the semantics of \andy in terms of timed
automata is equivalent to the semantics in terms of high-level Petri
nets.

\paragraph{\bf Timed automata.}

A timed automaton is an annotated directed (and connected) graph, 
with an initial node and provided with a finite set of non-negative real 
variables called \emph{clocks}. 
Nodes (called \emph{locations}) are annotated with \emph{invariants} 
(predicates allowing to enter or stay in a location).
Arcs are annotated with \emph{guards}, 
\emph{communication labels}, and possibly with some clock \emph{resets}. 
Guards are conjunctions of elementary predicates of the form  
$x~\op~c$, where $\op\in\{>,\geq,=,>,\leq\}$ where $x$ is a clock and 
$c$ a (possibly parameterised) positive integer constant.
As usual, the empty conjunction is interpreted as true. 
The set of all guards and invariant predicates will be denoted by $G$. 


\begin{definition}\label{def:ta} A \emph{timed automaton} $\TA$
is a tuple $(L,l^0,X,\Sigma,\arcs,\inv)$, where 
\begin{itemize}
	\item $L$  is a set of locations with $l^0\in L$ the initial one, $X$ is the set of clocks, 
	\item $\Sigma=\Sigma^s\cup\Sigma^u\cup\Sigma^b$ is a set of communication labels, 
	where $\Sigma^s$ are synchronous, $\Sigma^u$ are synchronous and urgent, 
	and $\Sigma^b$ are broadcast ones,
	\item $\arcs \subseteq L \times (G \cup \Sigma \cup R) \times L$ 
is a set of arcs between locations with a guard in $G$, 
a communication label in $\Sigma\cup\{\epsilon\}$, and a set of clock resets in $R=2^X$;
for all $a\in\Sigma^u$, we require the guard to be $\true$; 
  \item $\inv: L \rightarrow G$ assigns invariants to locations. 
\end{itemize}
\end{definition}


It is possible to define a synchronised product of a set of timed automata 
that work and synchronise in parallel. 
The automata are required to have disjoint sets of locations, but
may share clocks and communication labels which are used for synchronisation. 
We define three communication policies: 
\begin{itemize}
	\item \emph{synchronous} communications through
		labels $a\in\Sigma^s$ that require all the automata having label $a$  
		to synchronise on $a$; 
	\item \emph{synchronous urgent} communications through 
		labels $u\in\Sigma^u$ that are synchronous as above but \emph{urgent} meaning that
		there will be no delay if transition with label $u$ can be taken;
  \item \emph{broadcast} communications through labels $b!,b?\in\Sigma^b$ meaning that
	a set of automata can synchronise if one is emitting;
	notice that, a process can always emit (\eg $b!$) 
	and the receivers ($b?$) must synchronise if they can.
\end{itemize}

The synchronous product $\TA_1\parallel \ldots \parallel \TA_n$ 
of timed automata, where for each $j\in[1,\ldots,n]$, 
$\TA_j = (L_j, l^0_j,X_j,\Sigma_j,\arcs_j,\inv_j)$ and all $L_j$ are 
pairwise disjoint sets of locations is
the timed automaton $\TA=(L, l^0,X,\Sigma,\arcs,\inv)$ such that:
\begin{itemize}
\item $L=L_1\times\ldots\times L_n$ and $l^0=(l^0_1,\ldots,l^0_n)$, 
$X=\bigcup_{j=1}^n X_j$, $\Sigma=\bigcup_{j=1}^n \Sigma_j$,
\item $\forall l=(l_1, \ldots, l_n)\in L\colon \inv(l) = \bigwedge_j \inv_j(l_j)$,
\item $\arcs$ is the set of arcs 
$(l_1, \ldots, l_n) \stackrel{g,a,r}{\longrightarrow} (l'_1, \ldots, l'_n)$
such that (where for each $a\in\Sigma^s\cup\Sigma^u, 
S_a=\{j\mid 1\leq j \leq n, a\in \Sigma^s_j\cup\Sigma^u_j\}$):
	for all $1\leq j \leq n$, if $j\not\in S_a$, then $l_j'=l_j$, otherwise
	there exist $g_j$ and $r_j$ such that 
	$l_j \stackrel{g_j,a,r_j}{\longrightarrow} l_j'\in E_j$; $g=\bigwedge_{j\in S_a} g_j$
	and $r=\bigcup_{j\in S_a} r_j$.
\end{itemize}

The semantics of a synchronous product $\TA_1\parallel \ldots \parallel \TA_n$
is that of the underlying timed automaton $\TA$
(synchronising on synchronous and broadcast communication labels) as recalled below, 
with the following notations. 
A location is a vector $l = (l_1, \ldots, l_n)$. 
We write $l[l'_j/l_j, j\in S]$ to denote the location $l$ in which the 
$j$th element $l_j$ is replaced by $l'_j$, for all $j$ in some set $S$.
A valuation is a function $\nu$ from the set of clocks to the non-negative reals. 
Let $\mathbb{V}$ be the set of all clock valuations, and $\nu_0(x) = 0$ 
for all $x \in X$.  
We shall denote by $\nu\vDash F$ the fact that the valuation $\nu$ 
satisfies (makes true) the formula $F$.
If $r$ is a clock reset, we shall denote by $\nu[r]$ 
the valuation obtained after applying clock reset $r\subseteq X$ to $\nu$; 
and if $d\in\mathbb{R}_{> 0}$ is a delay, 
$\nu+d$ is the valuation such that, for any clock $x\in X$, 
$(\nu+d)(x)=\nu(x)+d$.

The semantics of a synchronous product $\TA_1\parallel \ldots \parallel \TA_n$
is defined as a timed transition system $(S,s_0,\rightarrow)$,
where $S = (L_1 \times,\ldots\times L_n) \times \mathbb{V}$ is the set of states, 
$s_0 = (l^0, \nu_0)$ is the initial state, and 
$\rightarrow \subseteq S \times S$ is the transition relation defined by:
\begin{itemize}	
	\item (sync): $(\bar{l},\nu) \rightarrow (\bar{l'},\nu')$ if 
	there exist arc $l \stackrel{g,a,r}{\longrightarrow} l'\in \arcs$ 
	such that $\nu\vDash g$, $\nu'=\nu[r]$,
	for $S_a=\{j\mid 1\leq j\leq n, a\in\Sigma^s_j\}$, 
	$l'=l[l'_j/l_j, j\in S_a]$, 
	and there is no enabled transition with urgent communication label
	from $(\bar{l},\nu)$;
	
	\item (urgent): as (sync) but $a\in\Sigma^u$ and there is no delay
	if transition with urgent communication label can be taken;
	 
	\item (broadcast): $(\bar{l},\nu) \rightarrow (\bar{l'},\nu')$ if
	there exist an output arc $l_j \stackrel{g_j,b!,r_j}{\longrightarrow} l_j'\in \arcs_j$ 
	and a (possibly empty) set of input arcs of the form
	$l_k \stackrel{g_k,b?,r_k}{\longrightarrow} l_k'\in \arcs_k$ such that 
	for all $k\in K=\{k_1,\ldots,k_m\}\subseteq\{l_1,\ldots,l_n\}\setminus\{l_j\}$,
	the size of $K$ is maximal, $\nu\vDash \bigwedge_{k\in K\cup\{j\}} g_k$,
	$l'=l[l'_k/l_k, k\in K\cup\{j\}]$ and $\nu'=\nu[r_k, k\in K\cup\{j\}]$;
	
	\item (timed): $(l,\nu)\rightarrow (l,\nu+d)$ if $\nu+d\vDash \inv(l)$. 
\end{itemize} 
Here we exemplify timed automata usage: 
consider for instance the network of timed automata $\TA_1$ and $\TA_2$ with synchronous (non urgent) communications only:
\begin{center}
\begin{tikzpicture}[>=latex',xscale=.8, yscale=.6,every node/.style={scale=0.7}]
\node[location,double] at (0,0) (l1) {$\stackrel{l_1}{x < 2}$}; 
\node[location]  at (4,0) (l2) {$\stackrel{l_2}{x < 2}$};
\node[left] at (-4,0) {$\TA_1$};
\node[location,double] at (0,-2) (l3) {$\stackrel{l_3}{\emptyset}$};
\node[location] at (4,-2) (l4) {$\stackrel{l_4}{\emptyset}$};
\draw[->, rounded corners] (l2) -- (5,-.5) -- (5.5,-.5) -- node[midway,right] 
				{$x>0;b;\emptyset$}  (5.5,.5) -- (5,.5)-- (l2) ;

\node[left] at (-4,-2) {$\TA_2$};
\draw[->] (l1) -- (l2) node[midway,above] {$x=1;a;\{x\}$};
\draw[->] (l3) -- (l4) node[midway,above] {$x=1;c;\emptyset$};
\draw[->, rounded corners] (l3) -- (-1,-2.5) -- (-1.5,-2.5) -- node[midway,left] 
				{$\true;a;\{x\}$}  (-1.5,-1.5) -- (-1,-1.5)-- (l3) ;
\end{tikzpicture}
\end{center}
whose behaviour is given by their synchronised product $\TA_1\parallel \TA_2$:
\begin{center}
\begin{tikzpicture}[>=latex',xscale=.9, yscale=.6,every node/.style={scale=0.7}]
\node[location]  at (0,0) (l14) {$\stackrel{(l_1,l_4)}{x < 2}$};
\node[location,double] at (3,0) (l13) {$\stackrel{(l_1,l_3)}{x < 2}$}; 
\node[location] at (6,0) (l23) {$\stackrel{(l_2,l_3)}{x < 2}$};
\node[location] at (9,0) (l24) {$\stackrel{(l_2,l_4)}{x < 2}$};
\draw[->, rounded corners] (l23) -- (5.5,1) -- (5.5,1.5) -- node[midway,above] 
				{$x>0;b;\emptyset$}  (6.5,1.5) -- (6.5,1)-- (l23) ;
\draw[->] (l13) -- (l14) node[midway,above] {$x=1;c;\emptyset$};
\draw[->] (l13) -- (l23) node[midway,above] {$x=1;a;\{x\}$};
\draw[->] (l23) -- (l24) node[midway,above] {$x=1;c;\emptyset$};
\draw[->, rounded corners] (l24) -- (8.5,1) -- (8.5,1.5) -- node[midway,above] 
				{$x>0;b;\emptyset$}  (9.5,1.5) -- (9.5,1)-- (l24) ;
\end{tikzpicture}
\end{center}
and where a possible run is:
\begin{center}
\begin{tikzpicture}[>=latex',xscale=1, yscale=.6,every node/.style={scale=0.7}]
\node at (0,0) (e1) {$[(l_1,l_3);x=0]$}; // t+1
\node at (2,0) (e2) {$[(l_1,l_3);x=1]$}; // a
\node at (4,0) (e3) {$[(l_2,l_3);x=0]$}; // t+.5
\node at (6,0) (e4) {$[(l_2,l_3);x=.5]$}; // b
\node at (8,0) (e5) {$[(l_2,l_3);x=.5]$}; // t+.5
\node at (10,0) (e6) {$[(l_2,l_4);x=1]$}; 
\draw[->] (e1) -- (e2) ;
\draw[->] (e2) -- (e3) ;
\draw[->] (e3) -- (e4) ;
\draw[->] (e4) -- (e5) ;
\draw[->] (e5) -- (e6) ;
\end{tikzpicture}
\end{center}

\paragraph{\bf Encoding into timed automata.}
We are now ready to introduce the encoding of the high-level Petri net formalisation of $\andy$, to this aim we need to add some notation:

 \begin{newnotation}
Let $\Tag=\{\beta_1,\ldots,\beta_n\}$ be the set of all \obligatory activities identifiers from $\Syn$,
ordered alphabetically, i.e., $\beta_i<\beta_j$ if $i<j$. 
We denote by 
\[
\Tag^{\circledast}=\{ seq(h) \mid h\in \mathcal{P}(\Tag) \wedge h\neq\emptyset \} \cup \{ \varepsilon\},
\] 
the set of sequences $seq(h)$ obtained by concatenating the identifiers in non-empty subsets $h$ of $\Tag$, where for each $h=\{\beta_{i_1},\ldots,\beta_{i_m}\mid \forall j,k:~ i_j<i_k \}\in \mathcal{P}(\Tag)$, $seq(h)=\beta_{i_1}\cdots \beta_{i_m}$. We assume that if $\Tag^{\circledast}=\{h_1 \mydots h_k\}$, then the $h_i$'s are ordered by decreasing length and alphabetically in such a way that $h_1=\beta_1\cdots\beta_{|\Tag|}$ and $h_k = \varepsilon$. 

Moreover, $\beta_i=h[i]$ is the identifier at the $i$-position, and  $h=h_1-h_2$ is  the sequence of identifier in $h_1$ without those in $h_2$.
  %
 \end{newnotation}

 The encoding of an $\andy$ network is the synchronised product of one timed automaton for each entity in \entities  together with a set of auxiliary automata that are used to handle \potential and \obligatory activities.  
 The idea is that the global state of an \andy network is divided into its local counterparts represented by state of entities (\ie their levels). Thus for each entity $\entity$ we build a timed automaton $\TA(\entity, \level{\entity})$ which has as many locations as the levels in $\entity$. Auxiliary automata are used to implement the encoding of places $p_{\rho}$ ($\TA(\rho)$ for $\rho \in \Syn \cup \Reac$) and to realise time progression together with \obligatory activities ($\TA_{\tick}$). 
 More formally:
 
 \begin{definition} Given an \andy network $(\entities, \Syn, \Reac)$, with initial expression level $\level{\entity}$ for each $\entity \in \entities$, the corresponding timed automata encoding is 
 \[\enc{(\entities, \Syn, \Reac)} = \prod_{\entity \in \entities} \TA(\entity,\level{\entity}) \parallel \prod_{\alpha \in  \Reac} \TA(\alpha) \parallel \prod_{\beta \in \Syn} \TA(\beta) \parallel \TA_{\tick} \]  
 where $\TA(\entity, \level{\entity})$, $\TA(\alpha)$,$\TA(\beta)$  and $\TA_{\tick}$ are defined next. In the following we assume $ \Tag $ to be the set of identifiers of \obligatory activities in $\Syn$. 
 \paragraph{\bf Entities.}
 $\TA(\entity, \level{\entity}) = (L_{\entity},l^0_{\entity},X_{\entity},\Sigma_{\entity}, \arcs_{\entity} ,\inv_{\entity})$
  where:
  \begin{itemize}
   \item $L_{\entity}=\{l^{\entity}_i \mid i \in [0 \mydots \setlev_{\entity} ]\} \cup \{k_i^h,k_i^{d,h} \mid i \in [0 \mydots \setlev_{\entity} ], h\in \Tag^{\circledast} \}$ with  $l^0_{\entity} = l^{\entity}_{\level{\entity}}$
   \item $X_{\entity}=\{\birth^{\entity}_i, \refr^{\entity}_i \mid i \in [0 \mydots \setlev_{\entity} ]\} \cup \{x_{\entity}\}$
   \item $\Sigma^s_{\entity} =  \{  \alpha \mid \alpha \text{ identifier of an activity in }  \Reac  \}$, $  \Sigma_{\entity}^b =  \Tag^{\circledast}$, $\Sigma_{\entity}^u =  \{\tick h \mid h \in \Tag^{\circledast} \} $ 
   \item $\arcs_{\entity} = \arcs_{\Reac} \cup \arcs_{\Syn}$ where 
$$
\arcs_{\Reac}\! =\! \{ l_j \xrightarrow{g(A_{\alpha}) \wedge g(I_{\alpha}) \wedge x_{\entity} = 0, \alpha, r} l_e \mid \entity_{A_{\alpha}} \leq j < \entity_{I_{\alpha}}, \alpha \in \Reac, \entity \in A_{\alpha} \cup I_{\alpha} \cup R_{\alpha} \}  
$$
with  $j, e, \entity_{A_{\alpha}}$, and $\entity_{I_{\alpha}}$ are levels of $\entity$ defined as follows:
    
    $$\entity_{A_{\alpha}} = \begin{cases}
                \level{a} & \text{if } (\entity, \level{a}) \in {A_{\alpha}}     \\
                0 & \text{otherwise}
               \end{cases}
\qquad 
\entity_{I_{\alpha}} = \begin{cases}
               \level{i} & \text{if } (\entity, \level{i}) \in {I_{\alpha}}\\
                \setlev_{\entity} & \text{otherwise}
               \end{cases}
$$

$$ g(A_{\alpha}) = \begin{cases}
\birth^{\entity}_{\level{a}} \geq \dur{\alpha} & \text{if } (\entity, \level{a}) \in A_{\alpha}\\
           \true & \text{otherwise}
          \end{cases}
\ 
g(I_{\alpha}) = \begin{cases}
\birth^{\entity}_{\level{i}} \geq \dur{\alpha} & \text{if } (\entity, \level{i}) \in I_{\alpha}\\
           \true & \text{otherwise}
          \end{cases}
$$

$$m = \begin{cases}
               \max(0, \min(j+v, \setlev_{\entity}-1)) & \text{if } (\entity, v) \in R_{\alpha} \\
                j & \text{otherwise}
               \end{cases}
$$

$$r = \begin{cases}
        \emptyset & \text{if } (\entity, v) \not \in R_{\alpha}  \\
        \{u^\entity_m,\} \cup \{\birth_x^{\entity} \mid x \in [j+1,m]\} & \text{if } (\entity, v) \in R_{\alpha}  \wedge m-j>0\\
        \{u^\entity_m,\} & \text{if } (\entity, v) \in R_{\alpha}  \wedge m-j=0\\
        \{u^\entity_m,\} \cup \{\birth_x^{\entity} \mid x \in [m+1,j]\} & \text{if } (\entity, v) \in R_{\alpha}  \wedge m-j<0
      \end{cases}
$$
$$
\begin{array}{ll}
 \arcs_{\Syn} = & \{ l_j \xrightarrow{g(d)\wedge g , h?, \emptyset} k_j^{d,h},  \mid  j \in [0\mydots\setlev_{\entity}-1], h\in \Tag^{\circledast}\} \ \cup \\
               & \{ l_j \xrightarrow{\neg g(d)\wedge g , h?, \emptyset} k_j^h,  \mid  j \in [0\mydots\setlev_{\entity}-1], h\in \Tag^{\circledast}\} \ \cup \\
              & \{ k_j^{d,h} \xrightarrow{\true, h', r \cup \{x_{\entity}\} } l_e,  \mid  j \in [0\mydots\setlev_{\entity}-1], h,h' \in \Tag^{\circledast}, h<h' \}\ \cup \\ 
              & \{ k_j^h \xrightarrow{\true, \tick h', r'\cup \{x_{\entity}\} } l_e',  \mid  j \in [0\mydots\setlev_{\entity}-1], h,h' \in \Tag^{\circledast}, h<h' \}
\end{array}
$$
where
$$\begin{array}{lcl}
   g & = & \! \bigwedge_{k=1}^{n} g(\beta_k)  \wedge \bigwedge_{k=1}^{m} \neg g(\beta'_m)\text{ for } h=\beta_1 \cdots \beta_n \text{ and } h_1- h= \beta'_1 \cdots \beta'_m \\
   g(\beta) & = & g'(A_\beta) \wedge g'(I_\beta) \text{ and } \beta= \activ{A_\beta}{I_\beta}{\dur{\beta}}{R_\beta}\\
   g(d) & = & u_j > \life_{\entity}(j)
  \end{array}
$$

$$ g'(A_\beta) = \begin{cases}
j \geq \level{a} \wedge \birth^{\entity}_{\level{a}} \geq \dur{\beta} & \text{if } (\entity, \level{a}) \in A_\beta\\
           \true & \text{otherwise}
          \end{cases}
$$
$$
g'(I_\beta) = \begin{cases}
j<\level{i} \wedge \birth^{\entity}_{\level{i}} \geq \dur{\beta} & \text{if } (\entity, \level{i}) \in I_\beta\\
           \true & \text{otherwise}
          \end{cases}
$$
$$
\begin{array}{lcl}
m & = & \max(0, \min(\sum_{i\in[1 \mydots n]} f(h[i]') +j -1, \setlev_{\entity}-1)) \\
m' & = & \max(0, \min(\sum_{i\in[1\mydots n]} f(h[i]')+j, \setlev_{\entity}-1)) \\
 & & \mbox{ where } f(h[i]) = \begin{cases}
        v & \text{if }  (\entity, v) \in R_{\beta_i} \\
        0 & \text{otherwise} 
       \end{cases}
\end{array}
$$
\[
\begin{array}{ll}
r = &\begin{cases}
        \{u^\entity_m\} \cup \{\birth_x^{\entity} \mid x \in [j+1,m]\} & \text{if }  m-j>0\\
        \{u^\entity_m\} & \text{if }   m-j=0 \\
        \{u^\entity_m\} \cup \{\birth_x^{\entity} \mid x \in [m+1,j]\} & \text{if }   m-j<0     
      \end{cases}
      \\
r' = &\begin{cases}
        \{u^\entity_{m'}\} \cup \{\birth_x^{\entity} \mid x \in [j+1,m']\} & \text{if }  m'-j>0\\
        \{u^\entity_{m'}\} & \text{if }  \entity \in h'  \wedge m'-j=0 \\
        \{u^\entity_{m'}\} \cup \{\birth_x^{\entity} \mid x \in [m'+1,j]\} & \text{if }   m'-j<0 \\
        \emptyset & \text{if }  \entity \notin h' 
      \end{cases}
\end{array}
     \]
where $\entity \in h$ denotes formula: $\exists \beta_k=\activ{A_{\beta_k}}{I_{\beta_k}}{\dur{\beta_k}}{R_{\beta_k}}$ s.t. $h=\beta_1\cdots \beta_n, 1\leq k \leq n \wedge (\entity, v) \in R_{\beta_k}$ 
    
   \item $Inv_{\entity}(l^{\entity}_i) = \refr^{\entity}_i \leq \life_{\entity}(i)$ for all $i \in [0 \mydots \setlev_{\entity}]$
   
  \end{itemize}

  \paragraph{\bf \Potential activity.}
$\TA(\alpha) = (L_{\alpha},l^0_{\alpha},X_{\alpha},\Sigma_{\alpha},\arcs_{\alpha},\inv_{\alpha})$  for $\alpha \in \Reac$ where
  \begin{itemize}
  \item $L_{\alpha}=\{l_\alpha\}$,  $l^0_{\alpha} = \alpha$, $X_{\alpha}= \{w_{\alpha}\}$,  $\Sigma^s_{\alpha} = \{ \alpha \} $
  \item $\arcs_{\alpha} = \{ l_\alpha \xrightarrow{w_{\alpha} \geq  \dur{\alpha}, \alpha, \{ w_{\alpha}\}} l_\alpha \}$
  \item $\inv_{\alpha}(l_{\alpha}) = \true$.
  \end{itemize}
  
  \paragraph{\bf \Obligatory activity.} $\TA(\beta) = (L_{\beta},l^0_{\beta},X_{\beta},\Sigma_{\beta},\arcs_{\beta},\inv_{\beta})$  for $\beta \in \Syn$ where
  \begin{itemize}
   \item $L_{\beta}=\{l_\beta, l'_\beta\}$,  $l^0_{\beta} = l_\beta$,  $X_{\beta}= \{w_{\beta}\}$,  $\Sigma_{\beta}^b =  \Tag^{\circledast}$,  $
\Sigma_{\beta}^u =  \{ \tick h \mid h\in \Tag^{\circledast}\} $
   \item $\arcs_{\beta} = \{ l_\beta \xrightarrow{w_{\beta} \geq  \dur{\beta}, h?, \emptyset} l'_\beta \mid h\in \Tag^{\circledast}, \beta\in h \} \cup \{ l'_\beta \xrightarrow{\true, \tick h, \{w_{\beta}\}} l'_\beta \mid h\in \Tag^{\circledast}, \beta\in h \}$
   \item $\inv_{\beta}(l_{\beta}) = \true$ and $\inv_{\beta}(l'_{\beta}) = \true$.
  
  \end{itemize}
  
  \paragraph{\bf Time.}  $\TA_{\tick} = (L_{\tick},l^0_{\tick},X_{\tick},\Sigma_{\tick},\arcs_{\tick},\inv_{\tick})$ where
  \begin{itemize}
   \item $L_{\tick}=\{l_h \mid h \in \Tag^{\circledast}\} \cup \{l_{\bot}\}$,  $l^0_{\tick} = l_{h_1}$,  $X_{\tick}= \{x\}$, $\Sigma_{\tick}^b =  \Tag^{\circledast}$, $ 
\Sigma_{\tick}^u =  \{ \tick h \mid h\in \Tag^{\circledast}\} $
   \item $\arcs_{\tick} =\begin{array}{l}
 \{l_{h_i} \xrightarrow{x=1, h_i!, \emptyset } l_{h_i+1}, l_{h_i+1} \xrightarrow{\true, \tick h_i, \{x\} } l_{h_1} \mid h,  i \in [1\mydots n-1]   \}  \\
\cup \ \{ l_{h_n} \xrightarrow{x=1, h_n!, \emptyset } l_{\bot}, l_{\bot} \xrightarrow{\true, \tick\varepsilon, \{x\} } l_{h_1}   \}
                  \end{array}
$
   \item $\inv_{\tick}(l) = \true$ for all $l \in L_{\tick}$.
  
  \end{itemize}

 \end{definition}

  \begin{theorem}
The above encoding of  \andy network $(\entities, \Syn, \Reac)$
is correct and complete.\looseness=-1
\end{theorem}
\begin{proof}[Sketch]
 Follows by induction on the length of the run and from a case analysis on the transition performed. 
 
Some intuitions on the proof follows. For each entity $\entity$, the corresponding marking of place $p_{\entity}$, $M(p_{\entity})=\tuple{\lev_{\entity}, \refr_{\entity}, \birth_{\entity}}$, in the Petri net representation is encoded by    
 the state (location $l^{\entity}_{\lev_{\entity}}$ and valuations of clocks variables $\refr^{\entity}_{\lev_{\entity}}, \birth^{\entity}_i$ for $i \in [0\mydots\setlev_{\entity} -1]$) of each timed automaton $\TA(\entity, \level{\entity})$. Marking of places $p_{\rho}$ (for $\rho \in \Syn \cup \Reac$) is given by the valuation of clock $w_{\rho}$ in the corresponding timed automaton $\TA(\rho)$.
 
 Each transition of the Petri net is encoded by (a series of) timed automata arcs.
 For each transition $t_{\alpha}$ (corresponding to  \potential activity  $\alpha$) involving $\entity$ there is a (synchronous) arcs in the timed automaton $\TA(\entity, \level{\entity})$ whose guard describes its role in the activity (activator, inhibitor or result). Clock $w_{\alpha}$ in   $\TA(\alpha)$ implements the constraint that the activity $\alpha$ is performed at most once in the interval $\dur{\alpha} $:  $w_{\alpha} \geq \dur{\alpha}$. 
 This way, the synchronous product of all automata reconstructs the full guard of the activity $\alpha$ and  exactly one  transition in the synchronised product of automata corresponds to the firing of transition $t_{\alpha}$. The state reached after this transition coincides with the corresponding marking in the Petri net.
 
 Transition $t_c$ is trickier as time progression causes  decay but more importantly the simultaneous action of \obligatory activities.
 Notice that \obligatory activities concern the global state of an \andy network (the maximal set of enabled \obligatory activities has to be performed each time $t_c$ fires) but each sub-automaton of the synchronised automaton as only a partial/local information.  
 That is why, we need to introduce the auxiliary automaton $\TA_{\tick}$ that coordinates and gathers partial information from all  other automata.
 Thus, the implementation of $t_c$ has two phases. The first one gathers partial information, performs the selection of the largest set of enabled \obligatory activities and forces the time to progress in a discrete fashion;  the second phase completes the time progression and synchronises all timed automata communicating the chosen maximal set of \obligatory activities.
 Both phases are initiated by   automaton $\TA_{\tick}$ which has two types of arcs: broadcast ones for the first phase and urgent synchronous ones for the second (see Figure \ref{fig:tick}).
\begin{figure}[t]
\centering
\begin{tikzpicture}[>=latex',xscale=1, yscale=.6,every node/.style={scale=0.7}]
\node[location,double] at (0,0) (l1) {$l_{\beta_1\beta_2}$}; 
\node[location]  at (2,0) (l2) {$l_{\beta_1}$};
\node[location]  at (4,0) (l3) {$l_{\beta_2}$};
\node[location]  at (6,0) (l4) {$l_{\epsilon}$};
\node[location]  at (8,0) (l5) {$l_{\bot}$};
\draw[->] (l1) -- node[midway,above] {$\beta_1\beta_2!$} (l2) ;
\draw[->] (l2) -- node[midway,above] {$\beta_1!$} (l3) ;
\draw[->] (l3) -- node[midway,above] {$\beta_2!$} (l4) ;
\draw[->] (l4) -- node[midway,above] {$\epsilon!$} (l5) ;
\draw[->, rounded corners] (l2) -- node[midway,right] {$\tick\beta_1\beta_2$}(2,1.5) -- (1,1.5) --  (l1) ;
\draw[->, rounded corners] (l3) -- node[midway,right] {$\tick\beta_1$}(4,2) -- (1,2) --  (l1) ;
\draw[->, rounded corners] (l4) -- node[midway,right] {$\tick\beta_2$}(6,2.5) -- (1,2.5) --  (l1) ;
\draw[->, rounded corners] (l5) -- node[midway,right] {$\tick\epsilon$}(8,3) -- (1,3) --  (l1) ;

\end{tikzpicture}
\caption{The shape of timed automaton $\TA_\tick$ for $\Tag^\circledast=\{\beta_1\beta_2,\beta_1,\beta_2,\epsilon\}$, where $\tick\beta_1\beta_2$, $\tick\beta_1$, $\tick\beta_2$, $\tick\epsilon$ are all synchronous urgent communication labels.}
\label{fig:tick}
\end{figure}
More precisely, $\TA_{\tick}$ progressively interrogates the entities timed automata $\TA(\entity, \level{\entity})$ and the \obligatory activities automata $\TA(\beta_i)$ to ``compute'' for each automaton the maximal set of enabled \obligatory activities. This is obtained through broadcast arcs labelled with sequences of \obligatory activities identifiers $h\in\Tag^\circledast$, from the longest ($h=seq(\Tag)=\beta_1\cdots\beta_n$) to the shortest ($h=\epsilon$, \ie no \obligatory activity is enabled).
 As broadcast is a non blocking transition and because of the ordering on sequences in $\Tag^\circledast$, each entity automaton chooses its maximal set of \obligatory activities it is involved in. If it is necessary, it also performs decay. When all automata $\TA(\entity, \level{\entity})$ and $\TA(\beta_{i})$ have agreed on some sequence $h=\beta_{i_1} \mydots \beta_{i_m}$ (in the worst case $h$ is empty) the first phase is completed and $\{ \beta_{i_1},  \mydots,  \beta_{i_m}\}$ is the largest set of enabled \obligatory activities. The second phase is then implemented with an urgent synchronous arc synchronising all automata: $\TA_{\tick}$, $\TA(\entity, \level{\entity})$, for each $\entity \in \entities$, and $\TA(\beta_i)$, for $i\in [i_1 \mydots i_m]$.
 Notice that guards on broadcast transitions constraint the clocks to progress by one time unit at once.   
 As a consequence, at the end of the two phase algorithm, timed automata of entities and of places $p_{\rho}$ together with the  corresponding clocks valuations exactly encode the marking reached after firing $t_c$. 
\end{proof}

\end{document}